\documentclass[11pt, onecolumn]{IEEEtran}
\pdfoutput=1
\usepackage[colorlinks=true,linkcolor=blue,urlcolor=black,bookmarksopen=true]{hyperref}
\usepackage{amsthm,amssymb,graphicx,multirow,amsmath,color,amsfonts}%,ulem}
\usepackage[update,prepend]{epstopdf}
\usepackage[noadjust]{cite}
\usepackage[latin1]{inputenc}
\usepackage{tikz}
\usetikzlibrary{arrows,calc}		% Optional ticks libraries
\usepackage{bbm} % for \mathbbm{1}
\usepackage{pdfpages}
\usepackage{tabulary}
\usepackage{multirow}
\usepackage{comment}
\usepackage{mathtools}
\usepackage{bigints}
\usepackage{enumerate}
\usepackage{slashbox}
\usepackage{xcolor}
\usepackage{subcaption}
\def\nbbE{{\mathbb{E}}}

\def\nbbR{{\mathbb{R}}}

\def\P{{\mathbb{P}}}
\def\nrmd{{\rm d}}

\def\nrmx{{\rm x}}

\def\calC{\mathcal{C}}
\def\calE{\mathcal{E}}

\def\nb0{{\mathbf{0}}}
\def\nb1{{\mathbf{1}}}

\def\nbbE{{\mathbb{E}}}

\def\nbbR{{\mathbb{R}}}

\def\nrmd{{\rm d}}

\def\nrmx{{\rm x}}

\newtheorem{lemma}{Lemma}

\newtheorem{definition}{Definition}

\newtheorem{theorem}{Theorem}

\newtheorem{remark}{Remark}

\def\P{\mathbb{P}}

\def\calC{\mathcal{C}}
\def\calE{\mathcal{E}}

\allowdisplaybreaks %
\begin{document}
\bstctlcite{IEEEexample:BSTcontrol}

%\title{Shortest Path Distance in Manhattan Poisson Line Cox Process%\thanks{Grants or other notes
%%about the article that should go on the front page should be
%%placed here. General acknowledgments should be placed at the end of the article.}
%}
%%\subtitle{Do you have a subtitle?\\ If so, write it here}
%
%%\titlerunning{Shortest Path Distance in MPLCP}        % if too long for running head
%
%\author{	Vishnu Vardhan Chetlur \and Harpreet S. Dhillon \and Carl P. Dettmann }
%
%%\authorrunning{Short form of author list} % if too long for running head
%
%\institute{V. V. Chetlur \at
%              Wireless@VT, Dept. of ECE, \\
%              Virginia Tech,\\
%              Blacksburg, VA, USA.\\
%              \email{vishnucr@vt.edu}           %  \\
%%             \emph{Present address:} of F. Author  %  if needed
%           H. S. Dhillon \at
%           Wireless@VT, Dept. of ECE, \\
%           Virginia Tech,\\
%           Blacksburg, VA, USA.\\
%           \email{hdhillon@vt.edu}            
%           C. P. Dettmann \at
%           School of Mathematics, \\
%           University of Bristol, UK.\\
%           \email{carl.dettmann@bris.ac.uk}  
%}


\title{Shortest Path Distance in Manhattan Poisson Line Cox Process}
\author{
	Vishnu Vardhan Chetlur, Harpreet S. Dhillon, Carl P. Dettmann
	\thanks{V. V. Chetlur and H. S. Dhillon are with Wireless@VT, Department of ECE, Virginia Tech, Blacksburg, VA (email: \{vishnucr, hdhillon\}@vt.edu). C. P. Dettmann is with School of Mathematics, University of Bristol, UK (email: carl.dettmann@bris.ac.uk). 
		  \hfill Manuscript updated: \today.}
}
\maketitle

\begin{abstract}
While the Euclidean distance characteristics of the Poisson line Cox process (PLCP) have been investigated in the literature, the analytical characterization of the path distances is still an open problem. In this paper, we solve this problem for the stationary Manhattan Poisson line Cox process (MPLCP), which is a variant of the PLCP. Specifically, we derive the exact cumulative distribution function (CDF) for the length of the shortest path to the nearest point of the MPLCP in the sense of path distance measured from two reference points: (i) the typical intersection of the Manhattan Poisson line process (MPLP), and (ii) the typical point of the MPLCP. We also discuss the application of these results in infrastructure planning, wireless communication, and transportation networks.
\end{abstract}


\section{Introduction}\label{sec:intro}

The study of random spatial patterns, formally called {\em stochastic geometry}, has played an important role in statistical physics. Some of the well known examples include the study of percolation over both lattices and {random sets of points, referred to as {\em point processes} \cite{barthelemy, aldous2, perc1, perc2}, as well as the characterization of the properties of tessellations formed by point processes and random sets of lines called {\em line processes} \cite{hilhorst, carl}.} In fact, as will be discussed {shortly}, the modern treatment of line processes was inspired by the study of particle trajectories in a cloud-chamber experiment \cite{goudsmit}. The Poisson line process (PLP), {which will be defined formally in Section \ref{sec:background}}, is often the preferred choice for analysis in this line of work due to its tractability \cite{calka, dodson2007planar, tykesson2012percolation}. Given its rich history, a lot is already known about the distributional properties of a PLP \cite{stoyan, miles11964}. However, there has been a growing interest in a doubly stochastic point process {that is constructed by defining a random set of points on each line of a PLP in $\nbbR^2$}, which is relatively less understood and is the focus of this paper. Specifically, we focus on the {distribution of the shortest distance between two points of this point process when traveling only along the random lines. This distance, which will henceforth be referred to as the {\em shortest path distance}, has not been analytically characterized in the literature yet.} Before formulating the problem mathematically, it is instructive to discuss the rich history of PLP and the context in which this new doubly stochastic point process has emerged. 

%\subsection{Historical perspective}
As mentioned earlier, the development of the theory of line process was inspired by a problem suggested to S. Goudsmit by N. Bohr, which dealt with the chance of intersection of the trajectories of the sub-atomic particles in a cloud-chamber experiment \cite{goudsmit}. This work has motivated a lot of research focused on the various properties of line processes in the latter half of the twentieth century. In particular, R. E. Miles authored a series of papers which explored the fundamental properties of the PLPs and the random polygons generated by the lines of the PLP \cite{miles11964, miles21964, miles1973}. Some of the other prominent works in the literature include the spectral analysis of line processes presented by M. S. Bartlett in \cite{bartlett3} and the study of higher order properties of stationary line processes by R. Davidson in \cite{davidson1970}. 

Owing to its analytical tractability, the PLP has found applications in material sciences \cite{fairclough}, image processing \cite{image_proc}, geology \cite{geomech}, telecommunication \cite{bacc_plp, volker1}, and localization networks \cite{aditya}. 
In \cite{fairclough}, the position of fibers in each layer of a fiber membrane is modeled using PLP to analyze the strength of the membranes. In \cite{bacc_plp}, F. Baccelli proposed to model the road system by a PLP to study the handover behavior in cellular networks. This spatial model was further used by V. Schmidt and his co-authors in the analysis of urban telecommunication networks \cite{volker, volker1, volker2, Gloaguen2006}. Modeling the obstacles in a cellular assisted localization network by a PLP, the blind spot probability of the typical target node is explored in \cite{aditya}. 

A Poisson line Cox process or a Cox process driven by PLP is a doubly stochastic point process constructed by populating points on the lines of a PLP such that the locations of points on each line form a {1D Poisson point process (PPP), which is formally defined in the next section}. The PLCP has recently been employed in several works pertaining to the analysis of vehicular communication networks \cite{vishnuJ2, vishnuJ4, vishnuL1, baccchoi, bacc_letter}. Unlike PLPs, the research on PLCPs is still in nascent stages as some of the fundamental properties of the PLCP have only been explored very recently. For instance, the distribution of various Euclidean distances between the points of the PLCP have been derived in \cite{vishnuJ2} to characterize the signal-to-interference plus noise ratio (SINR)-based coverage probability of the typical vehicular node in the network. The Laplace functional of the PLCP is provided in \cite{bacc_letter} and the asymptotic behavior of the PLCP is studied in \cite{vishnuJ4}. However, these works have only focused on the Euclidean distance properties of the PLCP due to their impact on the network performance. 

%\subsection{Contribution}

Although sparse, a few works in the literature have also explored the path distance characteristics of the PLCP. The authors of \cite{volker1} have analyzed the mean shortest path length between a point of the PLCP and its closest point from another Cox process on the same PLP in the sense of Euclidean distance. The asymptotic behavior of this shortest path distance was investigated in \cite{volker3}. However, the analytical characterization of these path distances is still an open problem in the literature and is the main contribution of this paper. For this purpose, we consider a special variant of the PLP called Manhattan Poisson line process (MPLP), which will be discussed in detail in the next section. {For a stationary Cox process constructed on the MPLP, referred to as MPLCP, we derive the exact CDF of the shortest path distance to the nearest point of the MPLCP in the sense of path distance from two reference points: (i) the typical intersection of the MPLP, and (ii) the typical point of the MPLCP.} To the best of our knowledge, this is the first work to present the analytical characterization of path distances in a MPLCP. We also discuss the utility of the path distance characteristics of the MPLCP in providing useful insights in the areas of wireless communications, transportation networks, urban planning, and personnel deployment.

\section{Background and Notation}\label{sec:background}
In this section, we present a brief introduction to line processes and some of its fundamental properties. While we discuss only those aspects of line processes that are necessary for this paper, a detailed account of the theory can be found in \cite{stoyan}. %First, we provide a brief description of the PPP which is a primary building block of our spatial model in the next subsection.

%%\subsection{Poisson point process}
%The PPP is one of the most basic and well-investigated point processes in the stochastic geometry literature. The formal definition of a PPP is discussed next.
%
%
%\begin{definition}
%	
%	\end{definiton}

\subsection{Line process preliminaries}
%While the theory of line processes can be defined in a $d$-dimensional space, we limit our discussion to $\nbbR^2$ in the context of this paper. 
{As the PPP is a primary building block in the construction of the MPLP considered in our paper, we begin our discussion by defining it formally next.}
\renewcommand{\labelitemi}{$\bullet$}
\begin{definition} 
	{{\em (Poisson point process.)} A random set of points $\Phi \subset \nbbR^d$ with intensity measure $\Lambda$ is a PPP if it satisfies the following two properties:}
	\begin{itemize}
		{\item The number of points of $\Phi$ within any bounded Borel set $A \subset \nbbR^d$, denoted by $N_p(A)$, follows a Poisson distribution, i.e.}
		\begin{align}
		\P (N_p(A) = k) = \frac{\exp(-\Lambda(A)) (\Lambda(A))^k}{k!},
		\end{align}
		{where $\Lambda(A)$ is the average number of points of $\Phi$ in $A$.
		\item The number of points of $\Phi$ lying in $n$ disjoint Borel sets form a set of $n$ independent random variables for arbitrary $n$, which is also termed the independent scattering property.}
	\end{itemize}
\end{definition}
{A PPP is said to be \textit{homogeneous} if it has a constant intensity $\lambda$, which is the average number of points per unit volume in $ \nbbR^d $.}

{\em Line process}. As mentioned in Section \ref{sec:intro}, a line process is just a random collection of lines. In order to define it more formally, first observe that any undirected line $L$ in $\nbbR^2$ can be uniquely parameterized by its signed perpendicular distance $\rho$ from the origin $o \equiv (0,0)$ and the angle $\theta$ subtended by the line with respect to the positive $x$-axis in counter clockwise direction, as shown in Fig. \ref{fig:prelim_r2}. The sign of $\rho$ is negative if the origin is to the right or above the line. Thus, the pair of parameters $\rho$ and $\theta$ can be represented as the coordinates of a point on the half-cylinder $\calC \equiv [0, \pi ) \times \nbbR $, which is termed as the {\em representation space}, as illustrated in Fig. \ref{fig:prelim_r2}. 
%Clearly, there is a one-to-one correspondence between the lines in $\nbbR^2$ and points on the cylindrical surface $\calC$.
Thus, a random collection of lines in $\nbbR^2$ can be constructed from a set of points on $\calC$. Such a set of lines generated by a PPP on $\calC$ is called a PLP. 
% A line process in $\nbbR^2$ is characterized by line density $\mu$ which is defined as the average line length per unit area. For a motion-invariant PLP generated by a homogeneous 2D PPP with density $\lambda$ in $\calC$, the relationship between $\mu$ and $\lambda$ is given by $\mu = \pi \lambda$. 

%Inspired by the road layouts of some large cities, such as New York, as well as the nature of mathematical problems to be studied in this paper, we limit our attention to a special case of PLP called {\em Manhattan Poisson line process (MPLP)}. Owing to its analytical tractability, this model has been used by different communities for modeling road layouts; please see~\cite{Math papers} for examples from the stochastic modeling community and~\cite{Wireless papers} for examples from the wireless communications community.

%vv% As mentioned earlier, we limit our discussion to the special case of PLP called MPLP in which the grid-like structure of lines resembles that of road networks in many cities. Also, due to its analytical tractability, this model has been widely used for spatial modeling of the road layouts in the wireless communications community, which served as the initial motivation behind this work \cite{HeathCorr,HeathMmwave}. In fact, the models used by the wireless standardization bodies, such as the third generation partnership project (3GPP), are simple variants or special instances of this model \cite{3gpp}. 
%%
As mentioned earlier, in this paper, we focus on the special case of PLP called MPLP in which the orientations of the lines are restricted to $\{0, \pi/2\}$, thereby obtaining a set of horizontal and vertical lines in $\nbbR^2$, as depicted in Fig. \ref{fig:prelim_r2}. Thus, the MPLP $\Phi_l$ in $\nbbR^2$ can be constructed from two independent 1D PPPs $\Psi_0$, and $\Psi_{\pi/2}$ along the lines $\theta = 0$, and $\theta = \pi/2$, respectively, in the representation space $\calC$. 
Alternatively, one can construct a MPLP by first populating points along the $x$ and $y$-axes in $\nbbR^2$ according to independent 1D PPPs $\Xi_x$ and $\Xi_y$ and drawing vertical and horizontal lines through those points, respectively. This interpretation is useful in visualizing some of the basic properties of MPLP which will be discussed next. In this paper, we will mainly follow this interpretation for the ease of clarity and exposition. 

{\em Stationarity}. Analogous to a point process, a line process $\Phi_l$ is stationary if the distribution of lines is invariant to any translation $T_{(t, \beta)}$, which corresponds to the translation of the origin by a distance $t$ in a direction that makes an angle $\beta$ with respect to positive $x$-axis in counter clockwise direction. Upon applying a translation $T_{(t, \beta)}$, the representation of a line $L$ in $\calC$ changes from $(\rho, \theta)$ to $(\rho + t \sin (\theta - \beta ), \theta )$. Therefore, a MPLP $\Phi_l$ is stationary if the 1D PPPs $\Psi_0$, and $\Psi_{\pi/2}$ are stationary or alternatively, $\Xi_x$ and $\Xi_y$ are stationary.

\begin{figure}
	\centering
	\begin{minipage}[t]{.50\textwidth}
		\centering
		\includegraphics[width=.95\textwidth]{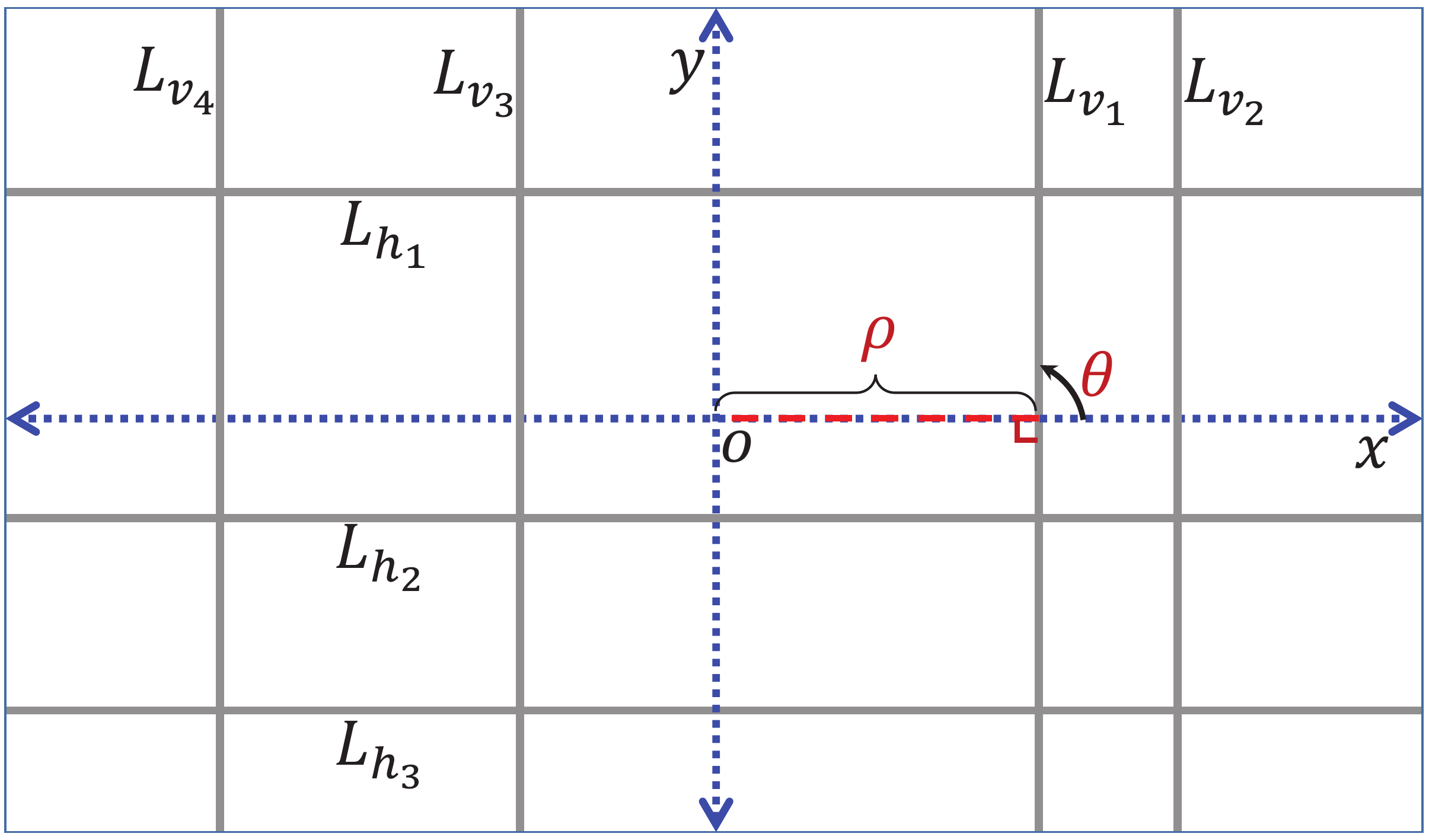}
%		\caption{Illustration of Manhattan Poisson line process in two-dimensional plane $\nbbR^2$.}
%		\label{fig:prelim_r2}
	\end{minipage}%
	\quad 
	\begin{minipage}[t]{.31\textwidth}
		\centering
		\includegraphics[width=.95\textwidth]{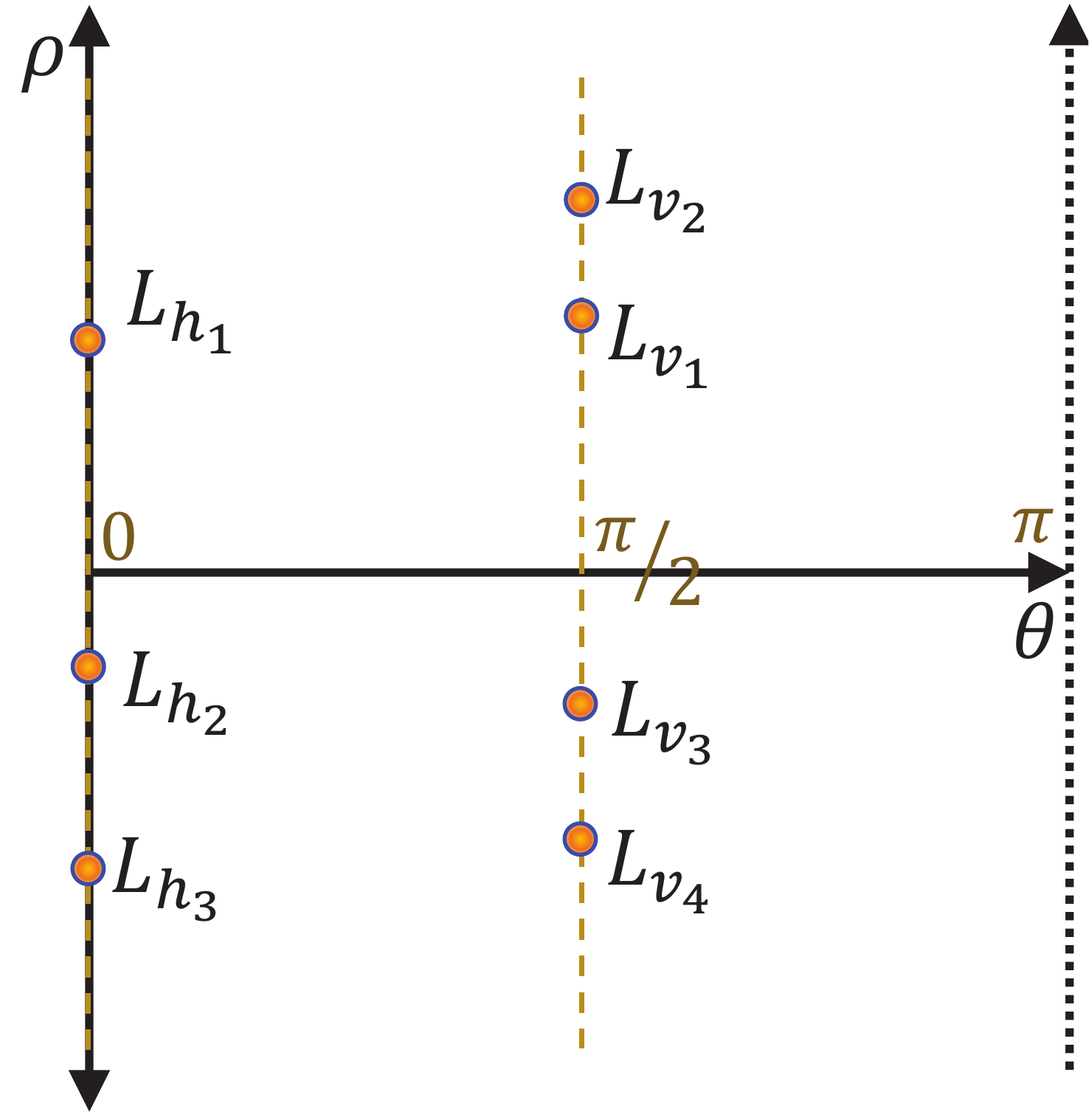}
%		\caption{Illustration of a point process in representation space $\calC \equiv [0, \pi) \times \nbbR $.}
%		\label{fig:prelim_c}
	\end{minipage}
	\caption{Illustration of the Manhattan Poisson line process in $\nbbR^2$ and the corresponding point process in representation space $\calC \equiv [0, \pi) \times \nbbR $.}
	\label{fig:prelim_r2}
\end{figure}

{\em Line density}. The line density of a line process is defined as the mean line length per unit area. The relationship between the line density and the density of the corresponding point process is given by the following Lemma.
\begin{lemma}
	For a stationary MPLP $\Phi_l$ constructed from independent and homogeneous 1D PPPs $\Xi_x$ and $\Xi_y$, each with density $\lambda_l$, the line density $\mu_l$ is given by $\mu_l = 2 \lambda_l$.
\end{lemma}
\begin{proof}
	Let us consider a ball of radius $d$ centered at the origin $b(o,d)$. We denote the set of horizontal and vertical lines of $\Phi_l$ by $\Phi_{lh}$ and $\Phi_{lv}$, respectively. The line density $\mu_l$ can now be computed as 
	\begin{align*}
\mu_l 	&= \frac{1}{\pi d^2} \nbbE \left[ \sum_{L \in \Phi_l} \nu_1( L \cap b(o, d) ) \right]\\
	&= \frac{1}{\pi d^2} \nbbE \left[ \sum_{L_h \in \Phi_{lh}} \nu_1( L_h \cap b(o, d) ) +  \sum_{L_v \in \Phi_{lv}} \nu_1( L_v \cap b(o, d) )\right]\\
	&= \frac{1}{\pi d^2} \nbbE \left[ \sum_{ \substack{\rho_x \in \Xi_{x} : \\ \rho_x \leq d}} 2 \sqrt{ d^2 - \rho_x^2} \right] + \frac{1}{\pi d^2} \nbbE \left[ \sum_{ \substack{\rho_y \in \Xi_{y} : \\\rho_y \leq d}} 2 \sqrt{ d^2 - \rho_y^2} \right]\\
	&\stackrel{(a)}{=} \frac{1}{\pi d^2} (\lambda_l 2 ) \bigg[ \int_0^1 2 \sqrt{ d^2-\rho_x^2} {\rm d}\rho_x +\int_0^1 2 \sqrt{ d^2-\rho_y^2} {\rm d}\rho_y \bigg] =2 \lambda_l, 
	\end{align*}
	where $\nu_1(\cdot)$ denotes the one dimensional Lebesgue measure and (a) follows from Campbell's theorem for sums over stationary 1D PPPs $\Xi_{x}$ and $\Xi_{y}$ \cite{haenggi}.
\end{proof}

{\em Lines intersecting a region}. For a stationary MPLP $\Phi_l$ with line density $\mu_l$, the number of horizontal and vertical lines that intersect a convex region $K \subseteq \nbbR^2$ are Poisson distributed with means $\mu_l \nu_1(K_y)/2$ and $\mu_l \nu_1(K_x)/2$, respectively, where $K_x$ and $K_y$ denote the projection of $K$ onto $x$ and $y$ axis.

\subsection{Spatial model and notation}
We will now provide a detailed description of the spatial model and also introduce the notation that will be followed in the paper. We consider a stationary MPLP $\Phi_l \equiv \{L_{h_1}, L_{h_2}, \dots, $ $ L_{v_1}, L_{v_2}, \dots \}$ in $\nbbR^2$ in which the vertical and horizontal lines are generated by independent homogeneous 1D PPPs $\Xi_x$ and $\Xi_y$, each having density $\lambda_l$. We denote the set of horizontal and vertical lines by $\Phi_{lh} \equiv \{L_{h_1}, L_{h_2}, \dots \}$ and $\Phi_{lv}\equiv \{L_{v_1}, L_{v_2}, \dots \}$, respectively. We construct a MPLCP $\Phi_c$ by populating points on the lines of $\Phi_l$ such that the locations of points on each line form a 1D PPP with density $\lambda_c$, as illustrated in Fig. \ref{fig:sysmod}. Note that the MPLCP $\Phi_c$ is also stationary due to the stationarity of the underlying MPLP and the homogeneity of 1D PPP on each line \cite{morlot, bacc_letter}. {As mentioned earlier in Section \ref{sec:intro}, we will consider two types of reference points with respect to which the path distance is measured: (i) the typical intersection of the MPLP, and (ii) the typical point of the MPLCP. As will be evident from the sequel, the treatment of the typical intersection case is relatively easier and will act as a precursor for the typical point case, whose exact treatment is the most important contribution of this paper.}

%V% As discussed already, some models used in industry can be argued to be special cases or variants of this model \cite{3gpp}. 
{For the typical intersection case, we measure the path distance of the nearest point of the MPLCP in the path distance sense with respect to the typical intersection of the MPLP, which can be placed at the origin $o$ without loss of generality. Thus, a horizontal line $L_x$ and a vertical line $L_y$, which are aligned along the $x$ and $y$-axes, respectively, pass through the typical intersection. In other words, both the homogeneous 1D PPPs $\Xi_x$ and $\Xi_y$ must now contain a point at the origin. By Slivnyak's theorem \cite{morlot, haenggi}, the conditioning on a point of the homogeneous PPP at the origin is equivalent to adding a point at the origin. Therefore, under this conditioning (more formally, under \textit{Palm probability}), the resulting line process is $\Phi_{l_0, {\rm int}} = \Phi_l \cup \{L_x, L_y\}$.} Thus, under the Palm probability of the intersection points, the resulting point process $\Phi_{c_0, {\rm int}}$ can be interpreted as the superposition of the point process $\Phi_c$ and two 1D PPPs each with density $\lambda_c$ along the lines $L_x$ and $L_y$ \cite{vishnuJ2, baccchoi}.

\begin{figure}
	\centering
	\includegraphics[width=.6\textwidth]{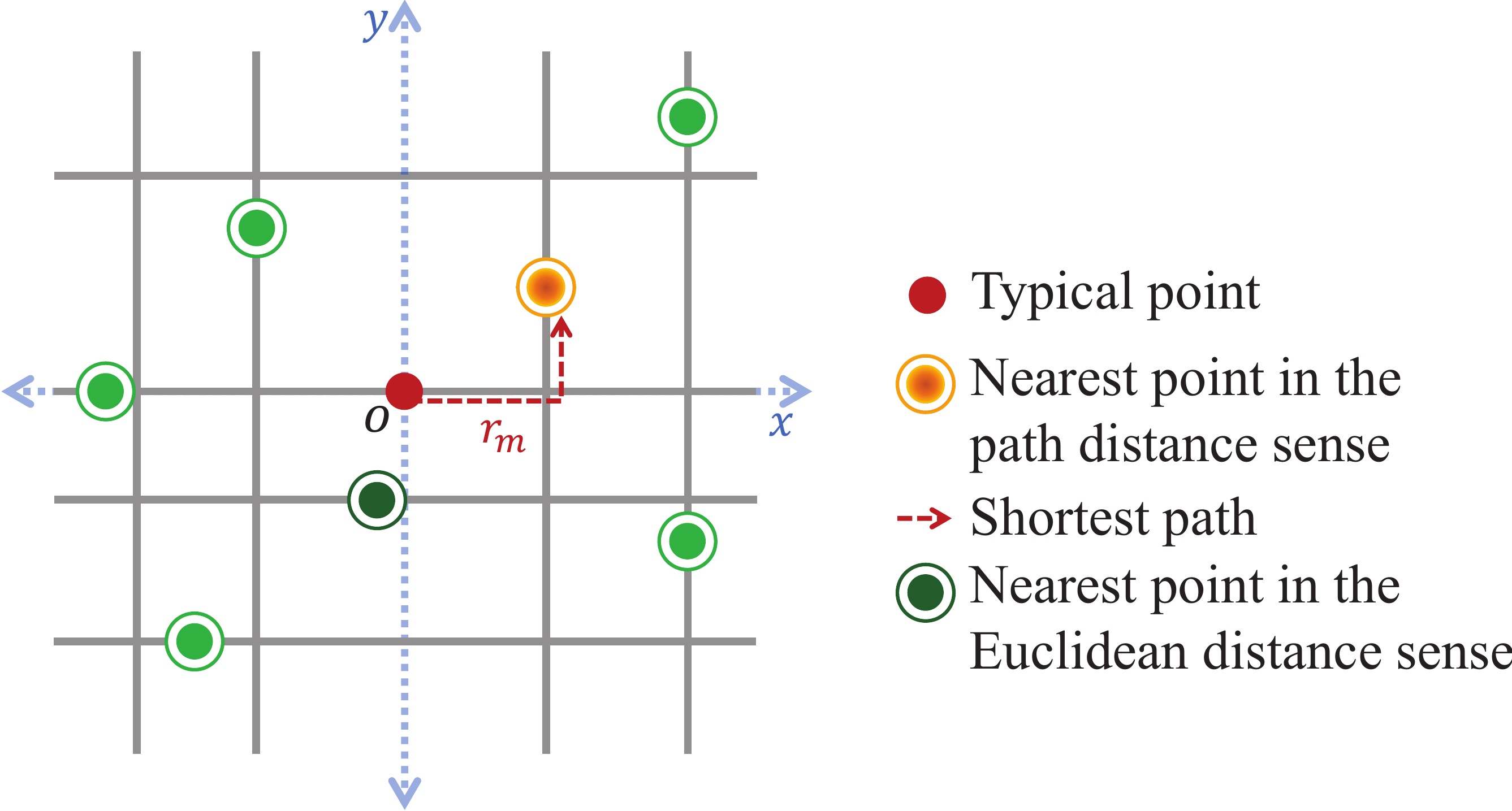}
	\caption{Illustration of the spatial model.}
	\label{fig:sysmod}
\end{figure}

{In case of the typical point of the MPLCP, without loss of generality, we assume that it is located on a horizontal line of the MPLP $\Phi_l$.} In this case, upon conditioning on the location of the typical point at the origin and using the same argument as above, the resulting line process is $\Phi_{l_0, {\rm typ}} = \Phi_l \cup \{L_x\}$. Thus, the resulting point process $\Phi_{c_0,{\rm typ}}$ can be interpreted as the superposition of the point process $\Phi_c$, an independent 1D PPP with density $\lambda_c$ on the line $L_x$ aligned along the $x$-axis and an atom at the origin. 

{Our main goal is to characterize the lengths of the shortest paths to the nearest point of the MPLCP in the sense of path distance for both the cases mentioned above.} We formally define the shortest path distance between two points as follows.
\begin{definition}
	{\em {(Shortest path distance.)}} The shortest path distance between two points ${\bf a}(x_1, y_1)$ and ${\bf b}(x_2, y_2)$ is defined as the sum of lengths of the line segments that constitute the shortest path $P$ from ${\bf a}$ to ${\bf b}$ and is denoted by $\ell({\bf a}, {\bf b})$.
\end{definition}

We denote the number of horizontal and vertical lines that intersect a region $A \subset \nbbR^2$ by $N_h(A)$ and $N_v(A)$, respectively. We denote the number of points of the MPLCP located in the set $A$ by $N_p(A)$. In this paper, we will denote the random variables by upper case letters and their corresponding realizations by lower case letters. For example, $W$ denotes a random variable, whereas $w$ denotes its realization. We defer the definition of other variables to later sections of the paper for better readability.

\section{Analytical Results} \label{sec:cdfrm_analytical}
{In this section, we will first characterize the distribution of the shortest path distance $T_m$ from the typical intersection of the MPLP to its nearest point of the MPLCP $\Phi_{c_0, {\rm int}}$ (under the palm distribution of the intersection points) in the sense of path distance. This will reveal a mathematical structure that will be useful later in the analysis of the shortest path distance measured with respect to the typical point of the MPLCP.}

\subsection{Shortest path distance from the typical intersection of the MPLP}
In this case, the length of the shortest path to any point located at $(x_i, y_i)$ is simply given by $z_i = | x_i| + |y_i|$, which is nothing but the first order Minkowski distance of the point from the origin. If the closest point of the MPLCP to the typical intersection (in the sense of path distance) is at a distance $t_m$, then there cannot be any point of the MPLCP at a location $(x,y)$ in $\mathbb{R}^2$ such that $|x| + |y| < t_m$. Thus, as depicted in Fig. \ref{fig:exclusionzone}, we obtain an exclusion zone $B_0$ formed by the intersection of the half-planes $x +y < t_m, -x + y < t_m, -x - y < t_m,$ and $x - y < t_m$. There can be no points on any of the line segments inside the square region $B_0$. In addition to $L_x$ and $L_y$, we know that there are a random number of lines that intersect the region $B_0$. From the construction of MPLP, it follows that the number of horizontal and vertical lines that intersect $B_0$ are Poisson distributed with mean $\lambda_l 2 t_m$. For a horizontal line located at a distance $y_l<t_m$ from the origin, the length of the line segment inside $B_0$ is given by $ 2t_m - 2y_l$. Similarly, for a vertical line at a distance $x_l<t_m$ from the origin, the length of the line segment inside $B_0$ is $ 2t_m - 2x_l$. Using these properties, we will now derive a closed-form expression for the CDF of the shortest path distance $T_m$ in the following theorem.

\begin{figure}
	\centering
	\includegraphics[width=.42\textwidth]{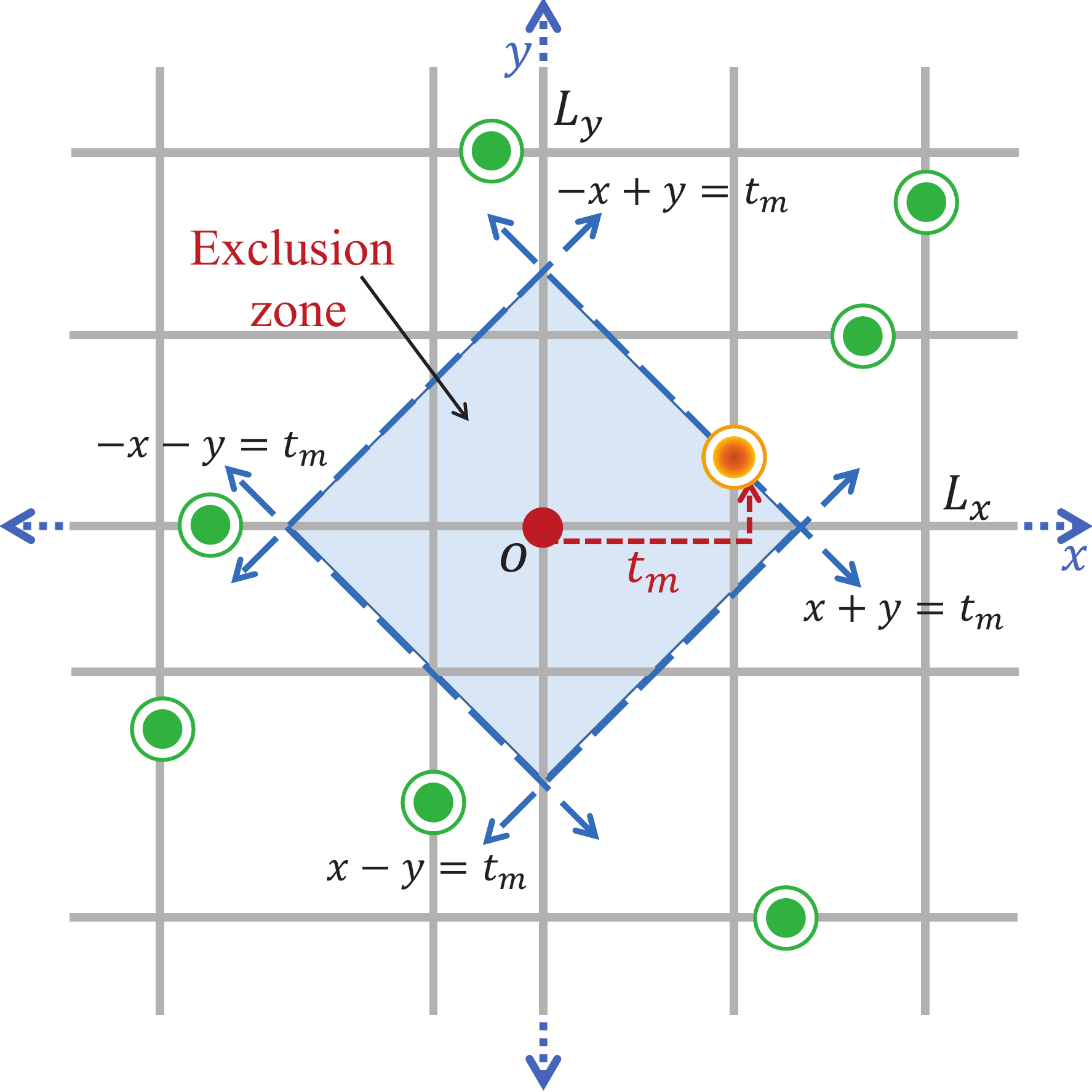}
	\caption{Illustration of the exclusion zone for the typical intersection.}
	\label{fig:exclusionzone}
\end{figure}

\begin{theorem}\label{theorem:rm_int_user}
	The CDF of the shortest path distance from the {typical intersection to its nearest point of the MPLCP} in the sense of path distance is 
	\begin{align} \label{eq:cdfrm_int}
	&F_{T_m}(t_m) = 1- \exp \left[ -4\lambda_c t_m -4\lambda_l t_m + \frac{2\lambda_l}{\lambda_c} \left(1 - e^{-2 \lambda_c t_m}\right)\right].
	\end{align}
\end{theorem}
\begin{proof}
	The CDF of $T_m$ can be computed as
	\begin{align*}
	F_{T_m}(t_m) &= 1 - \P ( T_m > t_m) \\
	&= 1 - \P( N_p(\Phi_{l_0, {\rm int}} \cap B_0) = 0)\\
	&\stackrel{(a)}{=} 1 - \P\Big( N_p\left( \{L_x \cup \Phi_{lh} \} \cap B_0 \right) = 0 \Big)  \P\Big(N_p\left( \{L_y \cup \Phi_{lv} \} \cap B_0 \right) = 0 \Big)
	\\	
	&\stackrel{(b)}{=} 1 - \bigg[\P(N_p(L_x\cap B_0) = 0) \sum_{n_{hl}=0}^{\infty} \P(N_{h}(B_0\setminus L_x) = n_{hl}) \P\Big( N_p\left(\Phi_{lh} \cap B_0 \right) = 0 | N_h(B_0\setminus L_x) = n_{hl} \Big) \bigg] \\
	& \hspace{2em} \times \bigg[\P(N_p(L_y\cap B_0) = 0) \sum_{n_{vl}=0}^{\infty} \P(N_{v}(B_0\setminus L_y) = n_{vl})  \P\Big( N_p\left( \Phi_{lv} \cap B_0 \right) = 0 | N_v(B_0\setminus L_y) = n_{vl} \Big) \bigg]
	\\	
	&\stackrel{(c)}{=} 1 - \Bigg[\P(N_p(L_x \cap B_0) = 0) \sum_{n_{hl}=0}^{\infty} \P(N_{h}(B_0\setminus L_x) = n_{hl}) \left( \prod_{j=1}^{n_{hl}}  \P\left( N_p(L_{h_j} \cap B_0 ) = 0\right) \right) \Bigg] \\
	&\hspace{4em} \times \Bigg[\P(N_p(L_y\cap B_0) = 0) \mspace{-6mu} \sum_{n_{vl}=0}^{\infty} \mspace{-5mu} \P(N_{v}(B_0\setminus L_y) \mspace{-5mu}= n_{vl}) \left( \prod_{k=1}^{n_{vl}}  \P\left( N_p(L_{v_k} \cap B_0 ) \mspace{-5mu} = 0\right) \right) \Bigg]
	\\
	&\stackrel{(d)}{=} 1- \Bigg[e^{-2 \lambda_c t_m} \sum_{n_{hl}=0}^{\infty} \frac{e^{-2 \lambda_l t_m }(2 \lambda_l t_m)^{n_{hl}}}{n_{hl}! }  \left(\int_0^{t_m} \exp\left(-\lambda_c(2t_m - 2y) \right) \frac{{\rm d}y}{t_m} \right)^{n_{hl}}\Bigg]\\
	&\hspace{4em} \times \Bigg[e^{-2 \lambda_c t_m} \sum_{n_{vl}=0}^{\infty} \frac{e^{-2 \lambda_l t_m }(2 \lambda_l t_m)^{n_{vl}}}{n_{vl}! }  \left(\int_0^{t_m} \exp\left(-\lambda_c(2t_m - 2x) \right) \frac{{\rm d}x}{t_m} \right)^{n_{vl}}\Bigg]
	\\
	&=1 - \left[ e^{- 2\lambda_c t_m} e^{-2 \lambda_l t_m} \exp\left[2 \lambda_l \int_0^{t_m} e^{-2\lambda_c (t_m - y)} {\rm d}y\right] \right] \left[ e^{- 2\lambda_c t_m} e^{-2 \lambda_l t_m} \exp\left[2 \lambda_l \int_0^{t_m} e^{-2\lambda_c (t_m - x)} {\rm d}x\right] \right]
	\\
	&= 1- \exp \left[ -4\lambda_c t_m -4\lambda_l t_m + \frac{2\lambda_l}{\lambda_c} \left(1 - e^{-2 \lambda_c t_m}\right)\right],
	\end{align*}
	where (a) follows from the fact that the distribution of horizontal and vertical lines are independent, (b) follows from conditioning on the number of horizontal and vertical lines intersecting the region $B_0$, (c) follows from the independent distribution of points on the lines, and (d) follows from the Poisson distribution of the number of lines intersecting $B_0$ and the void probability of 1D PPP on each line.
\end{proof}

\subsection{{Shortest path distance from the typical point of the MPLCP}}

In this subsection, we derive the exact CDF of the shortest path distance $R_m$ from the typical point of the MPLCP located at the origin to its nearest neighbor (nearest point of the MPLCP) in the sense of path distance. 

\begin{figure}
	\centering
	\includegraphics[width=.47\textwidth]{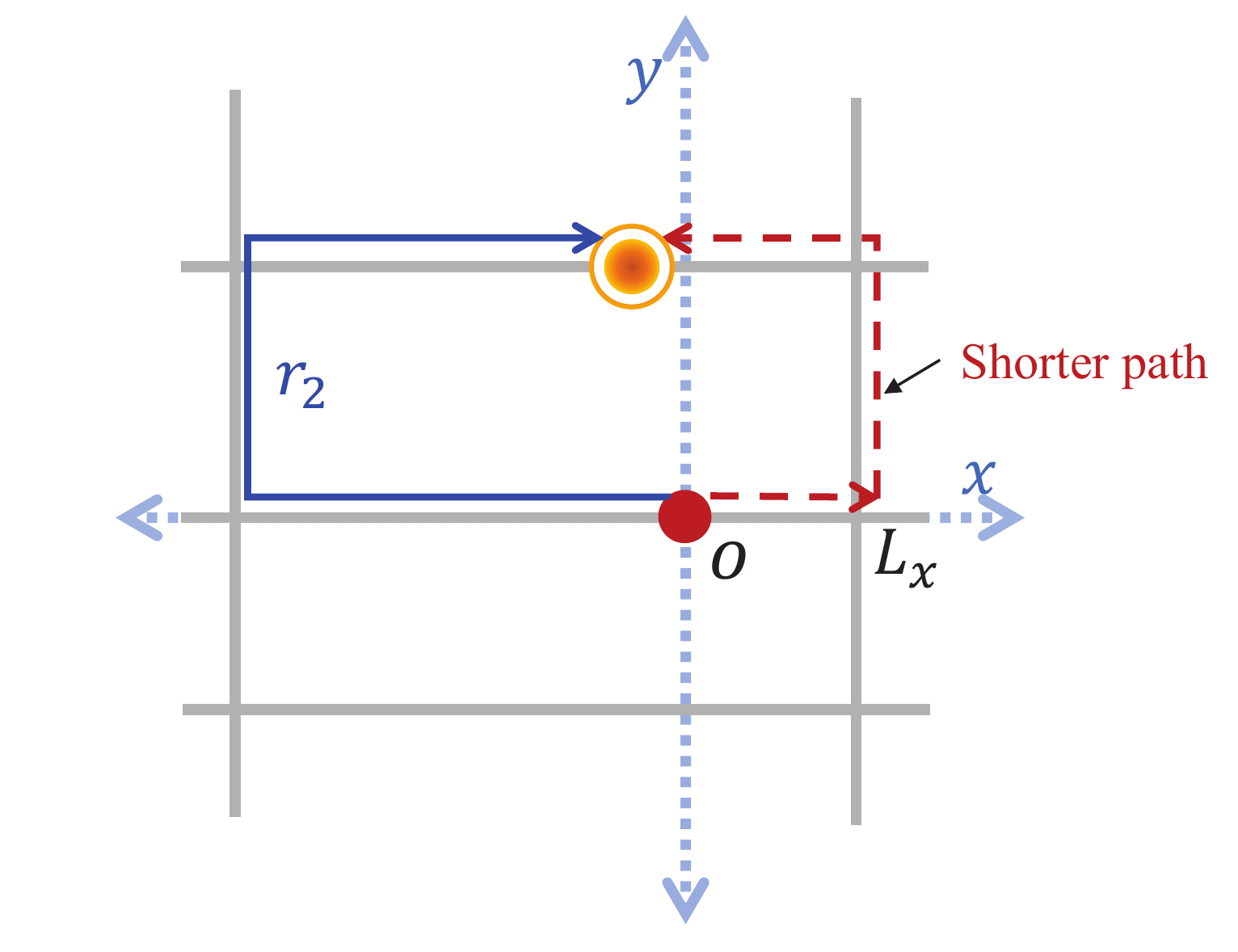}
	\caption{Illustration of the scenario in which the shortest path to a point to the left of
		the $y$-axis is the one that starts in the direction of positive $x$-axis.}
	\label{fig:exception}
\end{figure}

\begin{remark}\label{rem:deletion}
	The key difference between the spatial setup in this case and the previous case is that there does not exist a line $L_y$ along the $y$-axis in this setup. More precisely, the point process is now viewed under the regular distribution of $\Xi_x$ (which is equivalent to its reduced Palm distribution) and the Palm distribution of $\Xi_y$. 
\end{remark}

%	The key difference between the current and the previous setup highlighted in Remark \ref{rem:deletion} induces two main technical challenges in the characterization of the shortest path distance. First of all, unlike the previous case, the shortest path distance from the origin to any point is not necessarily equal to the first order Minkowski distance of the point. An example of this scenario is illustrated in Fig. \ref{fig:exception}. Second, the spatial 

	We will now discuss the technical challenges involved in the characterization of the shortest path distance in this setting as a result of the key difference highlighted in Remark \ref{rem:deletion}. First of all, the shortest path distance from the origin to some of the points is greater than the first order Minkowski distance of the point. In particular, these points are the ones that are located on the horizontal lines (excluding the line $L_x$) between the nearest vertical lines to the origin on either side, as illustrated in Fig. \ref{fig:exception}. Further, in some scenarios, the shortest path to some of the points that are located to the left of the $y$-axis is the one that starts in the direction of the positive $x$-axis and vice versa. This is also illustrated in Fig. \ref{fig:exception}. We will address these challenges in our analysis and derive the exact CDF of the shortest path distance. We will introduce some key variables that will be used in our analysis next.
	
	Let us denote the distance to the nearest vertical line (or the nearest intersection) to the right and the left of the origin by $S_r$ and $S_l$, respectively. From the construction of the MPLP, the PDFs of $S_i$, $i \in\{ l, r\}$, are given by
	\begin{align}
	f_{S_i} = \lambda_l \exp(-\lambda_l s_i), \quad 0 \leq s_i < \infty.
	\end{align}
	We denote the distance of the closer and the farther of the two nearest intersections on either side of the origin by $X_1$ and $X_2$, respectively. Without loss of generality, we consider the closer intersection to be on the positive direction of the $x$-axis for the rest of our discussion. We denote the location of the nearest intersection to the right and left of the origin by $\nrmx_1$ and $\nrmx_2$, respectively. As $X_1 = \min \{S_r, S_l \}$ and $X_2 = \max \{S_r, S_l\}$, their marginal PDFs are given by
	\begin{align}
	f_{X_1} (x_1) &= 2 \lambda_l \exp(-2 \lambda_l x_1), \\
	f_{X_2} (x_2) &= 2 \lambda_l \exp(-\lambda_l x_2) ( 1- \exp(-\lambda_l x_2) ).
	\end{align}
	The joint PDF of $X_1$ and $X_2$ is given by
	\begin{align}\label{eq:pdf_x1x2}
	f_{X_1, X_2} ( x_1 , x_2) = 2 \lambda_l^2 \exp(-\lambda_l (x_1 + x_2)), \quad 0 \leq x_1 \leq x_2 < \infty.
	\end{align}
	We now denote the distance of the nearest point of the MPLCP on $L_x$ from the origin in the direction of $X_1$ and $X_2$ by $D_1$ and $D_2$, respectively. The CDF and PDF of $D_j$, $j \in \{1, 2\}$, are
	\begin{align}
	F_{D_j} = 1 - \exp (-\lambda_c d_j), \\
	f_{D_j}(d_j) = \lambda_c \exp(-\lambda_c d_j).
	\end{align}
	
	Based on these random distances, there are now four possibilities: (i) $D_1 \leq X_1$, $D_2 > X_2$, (ii) $D_1 \leq X_1$, $D_2 \leq X_2$, (iii) $D_1 > X_1$, $D_2 > X_2$, and (iv) $D_1 > X_1$, $D_2 \leq X_2$. We denote these four events by $\calE_1$, $\calE_2$, $\calE_3$, and $\calE_4$, respectively. We will now compute the CDF of $R_m$ conditioned on each of these events and the distances $X_1$ and $X_2$. We will then obtain the overall CDF of $R_m$ using the law of total probability and taking expectation w.r.t. $X_1$ and $X_2$ in the last step.
	
	\subsubsection{Conditioned on $\calE_1$}\label{sec:e1}
	The occurrence of the event $\calE_1$ implies that the distance of the closest point on $L_x$ to the right of the origin is smaller than the distance to the nearest intersection in that direction and also the distance of the nearest point of the MPLCP on $L_x$ in the negative direction of the $x$-axis is greater than the distance of the nearest intersection in that direction. Thus, the probability of occurrence of the event $\calE_1$ conditioned on $X_1$ and $X_2$ can be computed as
	\begin{align}
	\notag \P \left( \calE_1 | X_1, X_2 \right) &= \P \left( D_1 \leq x_1, D_2 > x_2 | X_1 , X_2\right)\\
	\notag & \stackrel{(a)}{=} \P( D_1 \leq x_1 ) \P( D_2 > x_2)\\
	&\label{eq:pe1_x1x2} = (1 - \exp( -\lambda_c x_1)) \exp(-\lambda_c x_2 ),
	\end{align}
	where (a) follows from the fact that $D_1$ and $D_2$ are independent of each other and are also independent of $X_1$ and $X_2$.
	
	Conditioned on the occurrence of the event $\calE_1$, the shortest path distance from the origin is equal to the distance of the nearest point of the MPLCP on $L_x$ in the direction of the positive $x$-axis, i.e. $R_m = D_1$. Thus, the CDF of $R_m$ conditioned on $\calE_1$, $X_1$, and $X_2$ can be computed as
	\begin{align}
	\notag F_{R_m}(r_m | \calE_1, x_1 , x_2 ) &= 1 - \P \left( R_m > r_m | \calE_1, X_1, X_2  \right)\\
	\notag &= 1 - \frac{1}{\P \left( \calE_1 | X_1, X_2 \right) } \P \left( R_m > r_m , \calE_1| X_1, X_2  \right) \\
	\notag &= 1 - \frac{1}{\P \left( \calE_1 | X_1, X_2 \right) } \P \left( D_1 > r_m,  D_1 \leq x_1, D_2 > x_2 | X_1, X_2  \right) \\
	\notag & = 1 - \frac{1}{\P \left( \calE_1 | X_1, X_2 \right) } \P \left( r_m < D_1 \leq x_1 \right) \P \left( D_2 > x_2 \right)\\
	\notag & = 1 - \frac{\left( \exp(-\lambda_c r_m) - \exp(-\lambda_c x_1) \right) \exp \left( -\lambda_c x_2 \right)}{(1 - \exp( -\lambda_c x_1)) \exp(-\lambda_c x_2 )}\\
	\label{eq:cdfrm_e1_x1x2} & = \begin{dcases}
	\frac{1 - \exp(-\lambda_c r_m)}{1 - \exp(-\lambda_c x_1)}, \quad &0 \leq r_m \leq x_1, \\
	1, & x_1 < r_m.
	\end{dcases}
	\end{align}
	
	\subsubsection{Conditioned on $\calE_2$}
	The occurrence of $\calE_2$ means that the distance of the nearest points of the MPLCP on $L_x$ on either side of the origin are smaller than the distances to the corresponding intersections. Thus, we obtain the probability of occurrence of $\calE_2$ conditioned on $X_1$ and $X_2$ as
	\begin{align}
	\notag \P \left( \calE_2 | X_1, X_2 \right) &= \P \left( D_1 \leq x_1, D_2 \leq x_2 | X_1 , X_2\right)\\
	\notag &= \P( D_1 \leq x_1 ) \P( D_2 \leq x_2)\\
	&\label{eq:pe2_x1x2} = (1 - \exp( -\lambda_c x_1)) (1 -\exp(-\lambda_c x_2 )).
	\end{align}
	
	Following the same procedure as in Section \ref{sec:e1}, the CDF of $R_m$ conditioned on $\calE_2$, $X_1$, and $X_2$ can be obtained as
	\begin{align}
	\notag F_{R_m}(r_m | \calE_2, x_1 , x_2 ) &= 1 - \P \left( R_m > r_m | \calE_2, X_1, X_2  \right)\\
	\notag &= 1 - \frac{1}{\P \left( \calE_2 | X_1, X_2 \right) } \P \left( R_m > r_m , \calE_2| X_1, X_2  \right) \\
	\notag &= 1 - \frac{1}{\P \left( \calE_2 | X_1, X_2 \right) } \P \left( \min \{D_1, D_2 \} > r_m,  D_1 \leq x_1, D_2 \leq x_2 | X_1, X_2  \right) \\
	\notag & = 1 - \frac{1}{\P \left( \calE_2 | X_1, X_2 \right) } \P \left( r_m < D_1 \leq x_1 \right) \P \left( r_m < D_2 \leq x_2 \right)\\
	\notag & = 1 - \frac{\left( \exp(-\lambda_c r_m) - \exp(-\lambda_c x_1) \right) \left( \exp(-\lambda_c r_m) - \exp(-\lambda_c x_2) \right)}{(1 - \exp( -\lambda_c x_1)) (1 - \exp( -\lambda_c x_2))}\\
	\label{eq:cdfrm_e2_x1x2} & = \begin{dcases}
	\frac{1 - e^{-2 \lambda_c r_m} - (1 - e^{-\lambda_c r_m})(e^{-\lambda_c x_1} + e^{-\lambda_c x_2} )}{(1 - e^{ -\lambda_c x_1}) (1 - e^{ -\lambda_c x_2 } )}, \quad &0 \leq r_m < x_1 \\
	1, & x_1 < r_m.
	\end{dcases}
	\end{align}

	\subsubsection{Conditioned on $\calE_3$} 
	In this case, the closest points of the MPLCP on the line $L_x$ on either side of the origin are farther than locations of the nearest intersections. Thus, the probability of occurrence of $\calE_3$ conditioned on $X_1$ and $X_2$ can be computed as
	\begin{align}
	\notag \P(\calE_3 | X_1, X_2) &= \P( D_1 > x_1, D_2 > x_2 | X_1, X_2) \\
	\notag &= \P (D_1 > x_1) \P (D_2 > x_2)\\
	\label{eq:pe3_x1x2}&= \exp \left( -\lambda_c (x_1 + x_2) \right).
	\end{align}
	
	\begin{figure}
		\centering
		\includegraphics[width=.47\textwidth]{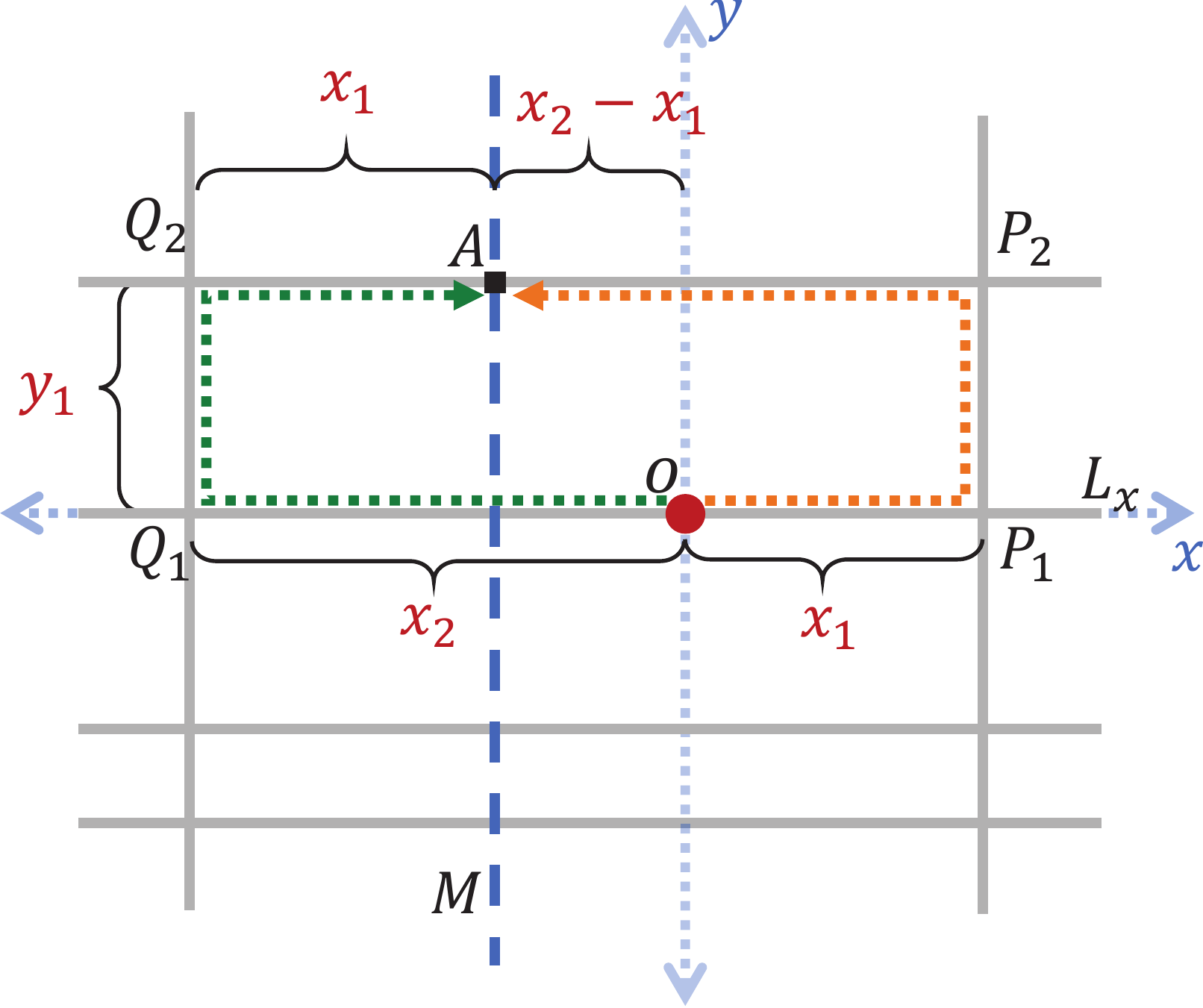}
		\caption{Illustration of the auxiliary line $M$ and the path distances from the origin to the intersection points of $M$ with the horizontal lines (except $L_x$).}
		\label{fig:construction}
	\end{figure}
	
	Recall that the key challenge in our analysis is the characterization of the shortest path distance for some of the points located to the left of the $y$-axis whose shortest path starts towards the positive $x$-axis and vice versa. This problem can be addressed by partitioning the points of the MPLCP into two sets based on the direction in which the shortest path to those points start from the origin. Conditioning on $X_1$ and $X_2$, we now divide the space into two half-planes using an auxiliary vertical line $M$ which is at a distance $(x_2 - x_1)$ to the left of the origin, as illustrated in Fig. \ref{fig:construction}. Let us now examine the shortest path to the points of intersection of this auxiliary line with the horizontal lines (excluding the line $L_x$). It can be observed that the path distances to these points obtained by starting to the right and the left of the origin are the same. For example, let us consider the intersection point $A$ shown in Fig. \ref{fig:construction}. Upon starting to the right of the origin, the shortest path distance to this point $A$ is $\overline{o P_1} + \overline{P_1 P_2} + \overline{P_2 A} = x_1 + y_1 + x_2$. Similarly, upon starting to the left of the origin, the shortest path distance to the point $A$ is $\overline{o Q_1} + \overline{Q_1 Q_2} + \overline{Q_2 A} = x_2 + y_1 + x_1$, which is equal to the shortest path distance when starting to the right of the origin. Therefore, the shortest path distances from the origin to the intersection points of the auxiliary line $M$ and the horizontal lines (except $L_x$) obtained by starting to the right and the left of the origin are the same. So, for all the points to the right of the line $M$, the shortest path from the origin starts to the right and likewise, for all the points to the left of this line, the shortest path starts towards the left of the origin.

	Based on the above construction, we now partition the points of the MPLCP into two sets based on the auxiliary line $M$ instead of the $y$-axis. So, we denote the shortest path distance to the origin upon starting to the right of the origin by $R_1 = X_1 + W_1$, where $W_1$ is the shortest path distance to the points located to the right of the auxiliary line $M$ from the intersection at $\nrmx_1$. Similarly, we denote the shortest path distance by starting to the left of the origin by $R_2 = X_2 + W_2$, where $W_2$ is the shortest path distance from the intersection at $\nrmx_2$ to the points located to the left of the auxiliary line $M$. Now, in order to compute the CDF of the overall shortest path distance, we need to determine the conditional CDFs of $W_1$ and $W_2$, which will be discussed next.

	The conditioning on $\calE_3$ already implies that there does not exist any point between the two intersections. This additional information about the distribution of points in the interval $(-x_2, x_1)$ on $L_x$ must be included in the computation of the conditional CDF of $W_1$. Similar to the procedure followed in the derivation of Theorem \ref{theorem:rm_int_user}, we will consider an exclusion zone $B$ formed by the intersection of the half-planes $(x-x_1)+y< w_1$, $-(x-x_1)+y<w_1$, $-(x-x_1)-y<w_1$, $(x-x_1)-y < w_1$, and $x > 0$. Note that the shape of the exclusion region $B$ depends on the values of $w_1$ with respect to $x_2$. While $B$ is a square for $w_1 \leq x_2$, it is a pentagon for $w_1 > x_2$, as shown in Figs. \ref{fig:w1_lt_x2} and \ref{fig:w1_gt_x2}, respectively. So, we will derive the conditional CDF of $W_1$ for the two cases $w_1 \leq x_2$ and $w_1 > x_2$ separately. We know that there cannot be any point on any of the line segments inside $B$. In addition to $L_x$, there exists a random number of horizontal lines above and below the line $L_x$ that intersect the region $B$. Likewise, in addition to the vertical line of the intersection $L_{v_0}$, there exists a random number of vertical lines that intersect the region $B$. However, conditioned on the event $\calE_3$, the distribution of vertical lines to the left of $L_{v_0}$ is not the same as the distribution of lines to the right of $L_{v_0}$. Since the first intersection to the right of the origin is at a distance $x_1$, there cannot be any vertical line that intersects $L_x$ in the interval $(-x_2, x_1)$, as shown in Fig. \ref{fig:w1_lt_x2}. So, we just need to focus on the set of vertical lines that intersect the region $B_{x_1^+} = B \cap \{x > x_1\}$. Due to symmetry, we obtain similar conditions for the computation of conditional CDF of $W_2$ as well. We will now derive the closed form expression for the conditional CDFs of $W_1$ and $W_2$ in the following Lemma.

	\begin{figure}
	\centering
	\includegraphics[width=.45\textwidth]{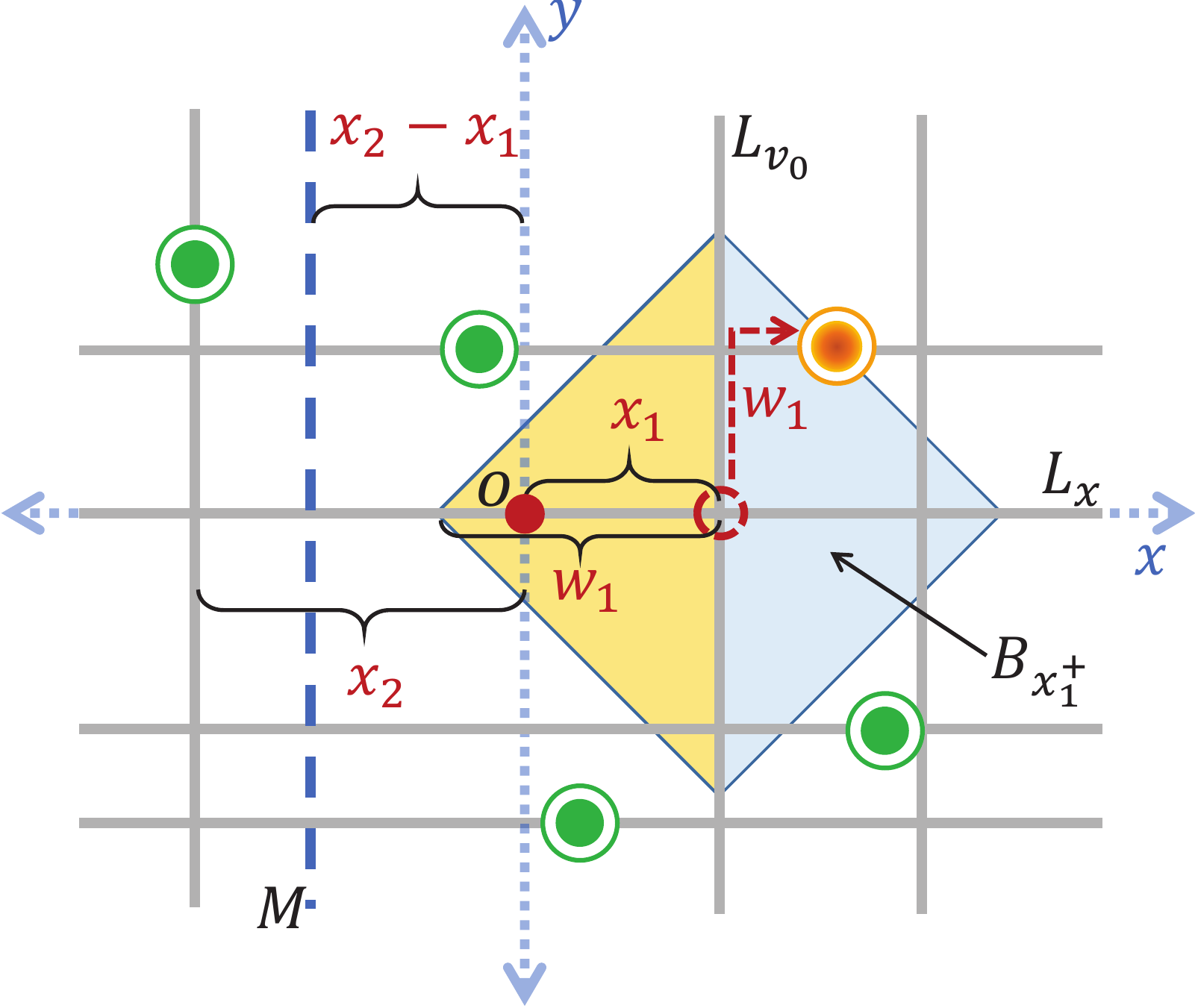}
	\caption{Illustration of the exclusion zone for the case $w_1 \leq x_2$.}
	\label{fig:w1_lt_x2}
	\end{figure}
	\begin{figure}
	\centering
	\includegraphics[width=.45\textwidth]{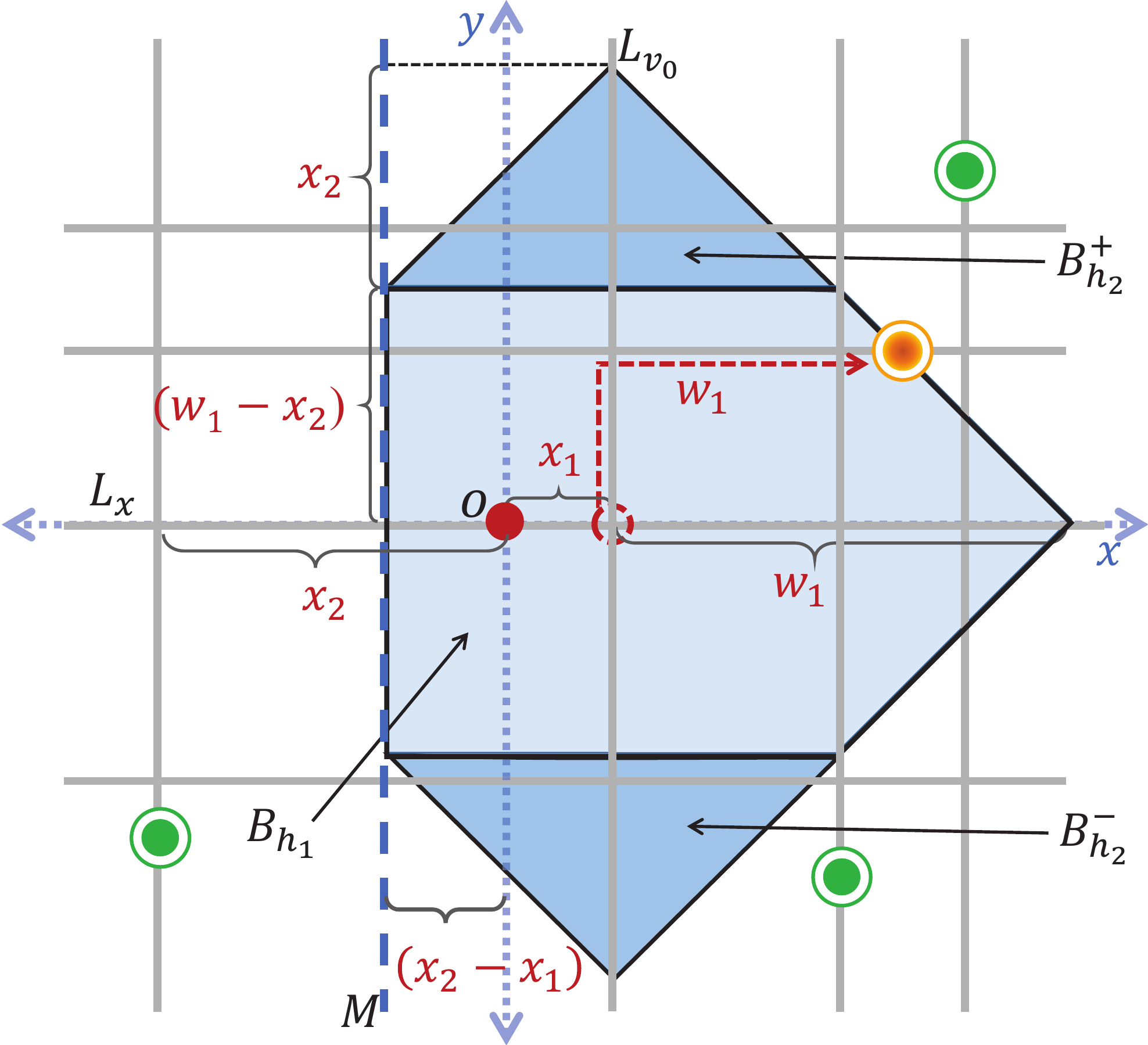}
	\caption{Illustration of the exclusion zone for the case $w_1 > x_2$.}
	\label{fig:w1_gt_x2}
	\end{figure}

\begin{lemma}\label{lem:cdfwm_e0}
	The CDFs of $W_1$ and $W_2$ conditioned on $\calE_3$, $X_1$, and $X_2$ are given by
	\begin{align}\label{eq:cdfwm_e0}
	F_{W_1}(w_1 | \calE_3 , x_1, x_2) = 
	\begin{dcases}
	F_{W_1, 1} (w_1 | \calE_3 , x_1, x_2), \quad 0 \leq w_1 \leq x_2 ,\\
	F_{W_1, 2} (w_1 | \calE_3 , x_1, x_2), \quad  w_1 > x_2, 
	\end{dcases}
	\end{align}
	where
	\begin{align} 
	F_{W_1, 1} (w_1 | \calE_3 , x_1, x_2)  &= 1 - \exp\Bigg[-3 \lambda_c w_1 - 3 \lambda_l w_1   +\frac{3 \lambda_l}{2 \lambda_c} \left(1 - e^{-2 \lambda_c w_1}\right)\Bigg],
	\end{align}
	and 
	\begin{align}
	F_{W_1, 2} (w_1 | \calE_3 , x_1, x_2)  = 1-\exp&\Bigg[- 3(\lambda_c + \lambda_l)w_1   +\frac{\lambda_l} {2\lambda_c} \big(3 + 2 e^{-2 \lambda_c x_2}- e^{-2 \lambda_c w_1} -4 e^{- \lambda_c (x_2+w_1)} \big)\Bigg],
	\end{align}
	
	\begin{align}\label{eq:cdfw2_e3}
	F_{W_2}(w_2 | \calE_3 , x_1, x_2) = 
	\begin{dcases}
	F_{W_2, 1} (w_2 | \calE_3 , x_1, x_2), \quad 0 \leq w_2 \leq x_1 ,\\
	F_{W_2, 2} (w_2 | \calE_3 , x_1, x_2), \quad  w_2 > x_1, 
	\end{dcases}
	\end{align}
	where
	\begin{align} 
	F_{W_2, 1} (w_2 | \calE_3 , x_1, x_2)  &= 1 - \exp\Bigg[-3 \lambda_c w_2 - 3 \lambda_l w_2   +\frac{3 \lambda_l}{2 \lambda_c} \left(1 - e^{-2 \lambda_c w_2}\right)\Bigg],
	\end{align}
	and 
	\begin{align}
	F_{W_2, 2} (w_2 | \calE_3 , x_1, x_2)  = 1-\exp&\Bigg[- 3(\lambda_c + \lambda_l)w_2  +\frac{\lambda_l} {2\lambda_c} \big(3 + 2 e^{-2 \lambda_c x_1}- e^{-2 \lambda_c w_2} -4 e^{- \lambda_c (x_1+w_2)} \big)\Bigg].
	\end{align}
\end{lemma}
\begin{proof}
	The conditional CDF of $W_1$ can be computed as 
	\begin{align}\notag
	F_{W_1}(w_1 | \calE_3, x_1, x_2) &= 1 - \P ( W_1 > w_1| \calE_3, X_1, X_2) \\ \label{eq:gen_gen_void_expr}
	&= 1 - \P( N_p(B) = 0| \calE_3, X_1, X_2).
	\end{align}
	As we had discussed earlier, the shape of the exclusion zone $B$ is different for the two cases $w_1 \leq x_2$ and $w_1 > x_2$ and hence we will handle these two cases separately. We will first consider the case $w_1 \leq x_2$. In this case, $B$ is a square region and we now need to determine the probability that there are no points inside this square region centered at an intersection, as shown in Fig. \ref{fig:w1_lt_x2}. By expressing the conditional void probability in \eqref{eq:gen_gen_void_expr} as the product of void probabilities of independent individual components, as in the proof of Theorem \ref{theorem:rm_int_user}, we obtain 
	\begin{align}\notag 
	\P&( N_p(B) = 0| \calE_3, X_1, X_2) \\ \notag
	&=\P(N_p(L_x\cap B) = 0 | \calE_3, X_1, X_2) \bigg[\sum_{n_{hl}=0}^{\infty} \P(N_{h}(B\setminus L_x) = n_{hl}| \calE_3, X_1, X_2) \\\notag
	&\hspace{1em}  \times \P\Big( N_p\left( \Phi_{lh} \cap B \right) = 0 | N_h(B\setminus L_x) = n_{hl}, \calE_3, X_1, X_2 \Big) \bigg] \\ \notag 
	&\hspace{1em} \times \P(N_p(L_{v_0}\cap B) = 0 | \calE_3, X_1, X_2) 
	\bigg[\sum_{n_{vl}=0}^{\infty} \P(N_{v}(B_{x_1^+}) = n_{vl}| \calE_3, X_1, X_2) \\ \notag 
	&\hspace{2em} \times \P\Big( N_p\left( \Phi_{lv} \cap B_{x_1^+} \right) = 0 | N_v(B_{x_1^+}) = n_{vl}, \calE_3, X_1, X_2\Big) \bigg]\\	 	\notag
	&\stackrel{(a)}{=} \P(N_p(L_x\cap B) = 0 | \calE_3, X_1, X_2) 
	\bigg[\sum_{n_{hl}=0}^{\infty} \P(N_{h}(B\setminus L_x) = n_{hl})  \P\Big( N_p\left(\Phi_{lh} \cap B \right) = 0 | N_h(B\setminus L_x) = n_{hl} \Big) \bigg]\\ \notag
	&\hspace{2em}\times\P(N_p(L_{v_0}\cap B) = 0)\bigg[ \sum_{n_{vl}=0}^{\infty} \P(N_{v}(B_{x_1^+}) = n_{vl})   \P\Big( N_p\left( \Phi_{lv}\cap B_{x_1^+} \right) = 0 | N_v(B_{x_1^+}) = n_{vl}\Big) \bigg]
	\\ \notag	
	&\stackrel{(b)}{=}  \P(N_p(L_x\cap B) = 0 | \calE_3, X_1, X_2)
	\Bigg[\sum_{n_{hl}=0}^{\infty} \P(N_{h}(B\setminus L_x) = n_{hl})  \left( \prod_{j=1}^{n_{hl}}  \P\left( N_p(L_{h_j} \cap B ) = 0\right) \right) \Bigg]\\ \notag
	&\hspace{2em} \times 
	\P(N_p(L_{v_0}\cap B) = 0) \Bigg[\sum_{n_{vl}=0}^{\infty} \P(N_{v}(B_{x_1^+}) = n_{vl}) 
	\left( \prod_{k=1}^{n_{vl}}  \P\left( N_p(L_{v_k} \cap B_{x_1^+} ) = 0 \right) \right) \Bigg]
	\\ \notag
	&\stackrel{(c)}{=} e^{- \lambda_c w_1}\Bigg[ \sum_{n_{hl}=0}^{\infty} \frac{e^{-2 \lambda_l w_1 }(2 \lambda_l w_1)^{n_{hl}}}{n_{hl}! } \left(\int_0^{w_1} \exp\left(-\lambda_c(2w_1 - 2y) \right) \frac{{\rm d}y}{w_1} \right)^{n_{hl}}\Bigg] \\ \notag
	&\hspace{3em} \times e^{-2 \lambda_c w_1}  \Bigg[ \sum_{n_{vl}=0}^{\infty} \frac{e^{- \lambda_l w_1 }( \lambda_l w_1)^{n_{vl}}}{n_{vl}! }
	\left(\int_0^{w_1} \exp\left(-\lambda_c(2w_1 - 2x) \right) \frac{{\rm d}x}{w_1} \right)^{n_{vl}}\Bigg]
	\\ \notag
	&= \left[ e^{- \lambda_c w_1} e^{-2 \lambda_l w_1} \exp\left[2 \lambda_l \int_0^{w_1} e^{-2\lambda_c (w_1 - y)} {\rm d}y\right] \right]  \left[ e^{- 2\lambda_c w_1} e^{- \lambda_l w_1} \exp\left[ \lambda_l \int_0^{w_1} e^{-2\lambda_c (w_1 - x)} {\rm d}x\right] \right]
	\\ \label{eq:voidprob_case1}
	&= \exp \left[ -3\lambda_c w_1 -3\lambda_l w_1 + \frac{3\lambda_l}{2\lambda_c} \left(1 - e^{-2 \lambda_c w_1}\right)\right],
	\end{align}
	where (a) follows from the fact that the distribution of points on the random horizontal lines, random vertical lines intersecting $B_{x_1^+}$, and the line $L_{v_0}$ is independent of $\calE_3$, $X_1$, and $X_2$, (b) follows from the independent distribution of points over different lines, and (c) follows from the Poisson distribution of lines and the void probability of 1D PPPs on those lines. 
	Substituting \eqref{eq:voidprob_case1} in \eqref{eq:gen_gen_void_expr}, we obtain the expression for the conditional CDF of $W_1$ for the case $w_1\leq x_2$.
	
	We will now consider the case $w_1 > x_2$, where the exclusion region $B$ is a pentagon as depicted in Fig. \ref{fig:w1_gt_x2}. The length of the horizontal line segment inside $B$ depends on the distance of the line from the origin. For a horizontal line $L_h$ which intercepts the $y$-axis at $y_h$ such that $|y_h| < w_1-x_2$, the length of the line segment inside $B$ is given by $x_2 + w_1-|y_h|$. On the other hand, if $|y_h| \geq w_1 - x_2$, then the length of line segment inside $B$ is $2(w_1 - |y_h|)$. So, we partition the set of horizontal lines that intersect $B$ into two sets: (i) the set of horizontal lines that intersect the region $B_{h_1} = B \cap \{|y|<w_1 - x_2\}$, and (ii) the set of horizontal lines that intersect the region $B_{h_2} = B \cap \{|y| \geq w_1 -x_2\}$. As $B_{h_2}$ is composed of two non-contiguous regions $B_{h_2}^+ = B \cap \{y \geq w_1 -x_2\}$ and $B_{h_2}^- = B \cap \{y \leq -(w_1 -x_2)\}$, we will handle them separately in our analysis. Thus, the conditional void probability for the case $w_1 > x_2$ can be computed as	
	\begin{align}\notag 
	\P&( N_p(B) = 0| \calE_3, X_1, X_2)\\ \notag 
	&=\P(N_p(L_x\cap B) = 0 | \calE_3, X_1, X_2) 
	\bigg[\sum_{n_{h_1}=0}^{\infty} \P\Big(N_{h} \big(B_{h_1} \setminus L_x \big) = n_{h_1}| \calE_3, X_1, X_2\Big) \\\notag
	&  \hspace{5em} \times 
	\P\Big( N_p\left( \Phi_{lh} \cap B_{h_1} \right) = 0 | N_h(B_{h_1} \setminus L_x) = n_{h_1}, \calE_3, X_1, X_2 \Big) \bigg] \\ \notag 
	&  \hspace{5em}\times 
	\bigg[\sum_{n_{h_2}=0}^{\infty} \P\Big(N_{h} \big(B_{h_2}^+ \big) = n_{h_2}| \calE_3, X_1, X_2\Big) 
	\P\Big( N_p\left( \Phi_{lh} \cap B_{h_2}^+ \right) = 0 | N_h(B_{h_2}^+ ) = n_{h_2},\calE_3, X_1, X_2 \Big) \\\notag
	& \hspace{5em} \times 
	\sum_{n_{h_3}=0}^{\infty} \P\Big(N_{h} \big(B_{h_2}^- \big) = n_{h_3}| \calE_3, X_1, X_2\Big) 	\P\Big( N_p\left( \Phi_{lh} \cap B_{h_2}^- \right) = 0 | N_h(B_{h_2}^- ) = n_{h_3},\calE_3, X_1, X_2 \Big) \bigg]\\\notag
	&  \hspace{5em} \times 
	\P(N_p(L_{v_0}\cap B) = 0 | \calE_3, X_1, X_2) 
	\bigg[\sum_{n_{vl}=0}^{\infty} \P(N_{v}(B_{x_1^+}) = n_{vl}| \calE_3, X_1, X_2) \\ \notag 
	& \hspace{5em} \times 
	\P\Big( N_p\left( \Phi_{lv} \cap B_{x_1^+} \right) = 0 | N_v(B_{x_1^+}) = n_{vl}, \calE_3, X_1, X_2\Big) \bigg]\\
	%%%%%%%%%%%%%%%%%%%%%%%%%%%%%%%%%%%%%%%%%%%%%%%%%%%%%%%%%%%%%%%%%%%%%%%%%%%%%%%%%%%	 
	\notag
	&\stackrel{(a)}{=}   \P(N_p(L_x\cap B) = 0 | \calE_3, X_1, X_2)  \P(N_p(L_{v_0}\cap B) = 0)\\
	\notag & \hspace{5em} \times 
	\bigg[\sum_{n_{h_1}=0}^{\infty} \P\Big(N_{h} \big(B_{h_1} \setminus L_x \big) = n_{h_1} | X_1, X_2\Big)  \P\Big( N_p\left( \Phi_{lh} \cap B_{h_1} \right) = 0 | N_h(B_{h_1} \setminus L_x) = n_{h_1}, X_1, X_2\Big) \bigg]\\
	\notag & \hspace{5em} \times 
	\bigg[\sum_{n_{h_2}=0}^{\infty} \P\Big(N_{h} \big(B_{h_2}^+ \big) = n_{h_2}| X_1, X_2\Big) \P\Big( N_p\left( \Phi_{lh} \cap B_{h_2}^+ \right) = 0 | N_h(B_{h_2}^+ ) = n_{h_2}, X_1, X_2 \Big) \\
	\notag & \hspace{5em} \times 
	\sum_{n_{h_3}=0}^{\infty} \P\Big(N_{h} \big(B_{h_2}^- \big) = n_{h_3}| X_1, X_2\Big)  \P\Big( N_p\left( \Phi_{lh} \cap B_{h_2}^- \right) = 0 | N_h(B_{h_2}^- ) = n_{h_3} , X_1, X_2 \Big) \bigg]\\
	\notag & \hspace{5em} \times
	\bigg[ \sum_{n_{vl}=0}^{\infty} \P(N_{v}(B_{x_1^+}) = n_{vl}|X_1, X_2)  \P\Big(N_p( \Phi_{lv} \cap B_{x_1^+} ) = 0 \big|  N_v(B_{x_1^+}) = n_{vl}, X_1, X_2\Big)  \bigg],\\ \notag
	%%%%%%%%%%%%%%%%%%%%%%%%%%%%%%%%%%%%%%%%%%%%%%%%%%%%%%%%%%%%%%%%%%%%%%%%%%
	&\stackrel{(b)}{=}  \P(N_p(L_x\cap B) = 0 | \calE_3, X_1, X_2)  \P(N_p(L_{v_0}\cap B) = 0)\\ \notag
	&\hspace{5em}\times \Bigg[ \sum_{n_{h_1}=0}^{\infty}\P(N_{h}(B_{h_1} \setminus L_x) = n_{h_1}| X_1, X_2) 
	\left( \prod_{i=1}^{n_{h_1}}  \P\left( N_p(L_{h_i} \cap B_{h_1} ) = 0 | X_1, X_2 \right) \right) \Bigg]\\ \notag
	&\hspace{5em}\times \Bigg[ \sum_{n_{h_2}=0}^{\infty}\P(N_{h}(B_{h_2}^+) = n_{h_2}|X_1, X_2) 
	\left( \prod_{j=1}^{n_{h_2}}  \P\left( N_p(L_{h_j} \cap B_{h_2}^+ ) = 0|X_1, X_2\right) \right) \\ \notag	
	&\hspace{5em}\times  \sum_{n_{h_3}=0}^{\infty}\P(N_{h}(B_{h_2}^-) = n_{h_3}|X_1, X_2) 
	\left( \prod_{j=1}^{n_{h_3}}  \P\left( N_p(L_{h_j} \cap B_{h_2}^- ) = 0|X_1, X_2\right) \right) \Bigg]\\ \notag
	&\hspace{5em} \times \Bigg[ \sum_{n_{vl}=0}^{\infty} \P(N_{v}(B_{x_1^+}) = n_{vl})
	\left( \prod_{k=1}^{n_{vl}}  \P\left( N_p(L_{v_k} \cap B_{x_1^+} ) = 0 \right) \right) \Bigg]
	\\
	%%%%%%%%%%%%%%%%%%%%%%%%%%%%%%%%%%%%%%%%%%%%%%%%%%%%%%%%%%%%%%%%%%%%%%%%%%%%%%%%% 
	\notag
	&\stackrel{(c)}{=}  e^{- \lambda_c w_1}  e^{-2 \lambda_c w_1} 
	\Bigg[ \sum_{n_{h_1}=0}^{\infty} \frac{e^{-2 \lambda_l (w_1-x_2) }(2 \lambda_l (w_1-x_2))^{n_{h_1}}}{n_{h_1}! } 
	\left(\int_0^{w_1-x_2} \mspace{-20mu} e^{-\lambda_c(x_2+w_1 - y) } \frac{{\rm d}y}{(w_1-x_2)} \right)^{n_{h_1}}\Bigg]\\
	 \notag
	&\hspace{5em} \times
	 \Bigg[ \sum_{n_{h_2}=0}^{\infty} \frac{e^{- \lambda_l x_2 }(\lambda_l x_2)^{n_{h_2}}}{n_{h_2}! }  \left(\int_{w_1-x_2}^{w_1} \mspace{-20mu}e^{-\lambda_c(2w_1 - 2y)  }\frac{{\rm d}y}{x_2} \right)^{n_{h_2}}\\
	\notag 
	&\hspace{5em} \times
	\sum_{n_{h_3}=0}^{\infty} \frac{e^{- \lambda_l x_2 }(\lambda_l x_2)^{n_{h_3}}}{n_{h_3}! } \left(\int_{w_1-x_2}^{w_1} \mspace{-20mu}e^{-\lambda_c(2w_1 - 2y)  }\frac{{\rm d}y}{x_2} \right)^{n_{h_3}}\Bigg] \\ \notag 
	& \hspace{5em}\times	
	\Bigg[\sum_{n_{vl}=0}^{\infty} \frac{e^{- \lambda_l w_1 }( \lambda_l w_1)^{n_{vl}}}{n_{vl}! }  \left(\int_0^{w_1} \mspace{-6mu} e^{-\lambda_c(2w_1 - 2x) } \frac{{\rm d}x}{w_1} \right)^{n_{vl}} \Bigg]\\
	%%%%%%%%%%%%%%%%%%%%%%%%%%%%%%%%%%%%%%%%%%%%%%%%%%%%%%%%%%%%%%%%%%%%%%%%%%%%%%%%%%%%%5
	\label{eq:voidprob_case2}
	&= \exp\Bigg[- 3(\lambda_c + \lambda_l)w_1  +\frac{\lambda_l} {2\lambda_c} \big(3 + 2 e^{-2 \lambda_c x_2} - e^{-2 \lambda_c w_1} -4 e^{- \lambda_c (x_2+w_1)} \big)\Bigg],
	\end{align}
	where (a) follows from the fact that the distribution of points on the random horizontal lines, vertical lines intersecting $B_{x_1^+}$ and the line $L_{v_0}$ are independent of $\calE_3$, (b) follows from the independent distribution of points over lines, and (c) follows from the Poisson distribution of lines and the void probability of 1D PPPs on each of those lines. Substituting \eqref{eq:voidprob_case2} in \eqref{eq:gen_gen_void_expr}, we obtain the expression for the conditional CDF of $W_1$ for the case $w_1 > x_2$. The CDF of $W_2$ conditioned on $\calE_3$, $X_1$, and $X_2$ can be obtained by following the same procedure. This completes the proof.
\end{proof}

Having determined all the components required to compute the CDF of $R_1$ and $R_2$ conditioned on $\calE_3$, $X_1$, and $X_2$, the conditional CDF of $R_m$ can now be computed as
\begin{align}
\notag & F_{R_m}\left( r_m | \calE_3, x_1, x_2 \right) = 1 - \P \left( R_m > r_m | \calE_3, X_1, X_2 \right)\\
%\notag &= 1 - \frac{1}{\P \left( \calE_3 | X_1, X_2 \right)} \P \left( R_m > r_m , \calE_3 |  X_1, X_2 \right) \\
\notag &= 1 -  \P \left( \min\{R_1, R_2\} > r_m  | \calE_3, X_1, X_2 \right)\\
\notag &= 1 - { \P \left( x_1 + W_1 > r_m , x_2 + W_2 > r_m |  \calE_3, X_1, X_2 \right)} \\
\notag &=1 - \left( 1 - F_{W_1} (r_m  - x_1 | \calE_3, x_1, x_2 ) \right)  \left( 1 - F_{W_2} (r_m  - x_2 | \calE_3, x_1, x_2  ) \right)\\
\notag &=\begin{dcases}
F_{W_1} ( r_m - x_1 | \calE_3, x_1, x_2 ), \quad & 0 \leq x_1 \leq r_m  < x_2,\\
F_{W_1} ( r_m - x_1 | \calE_3, x_1, x_2  ) + F_{W_2} ( r_m - x_2 | \calE_3, x_1, x_2 ) &\\
\hspace{3em} - F_{W_1} ( r_m - x_1 | \calE_3, x_1, x_2 )F_{W_2} ( r_m - x_2 | \calE_3, x_1, x_2 ),  \  & x_2 \leq r_m < \infty 
\end{dcases}\\
&\label{eq:cdfrm_e3_x1x2}=\begin{dcases}
F_{W_1,1} (r_m - x_1| \calE_3, x_1, x_2 ), &0 \leq x_1 \leq r_m  < x_2,\\
F_{W_1,1} (r_m - x_1| \calE_3, x_1, x_2 ) + F_{W_2,1}(r_m - x_2| \calE_3, x_1, x_2 )  &\\
\hspace{3em} - F_{W_1,1} (r_m - x_1| \calE_3, x_1, x_2 ) F_{W_2,1}(r_m - x_2| \calE_3, x_1, x_2 ), \ & x_2 \leq r_m < x_1 + x_2, \\
F_{W_1,2} (r_m - x_1| \calE_3, x_1, x_2 ) + F_{W_2,2}(r_m - x_2| \calE_3, x_1, x_2 )  &\\
\hspace{3em} - F_{W_1,2} (r_m - x_1| \calE_3, x_1, x_2 ) F_{W_2,2}(r_m - x_2| \calE_3, x_1, x_2 ), \ & x_1 + x_2 \leq r_m < \infty.
\end{dcases}
\end{align}

\subsubsection{Conditioned on $\calE_4$}
In this case, the nearest point on $L_x$ on the right side of the origin is farther than the corresponding intersection $\nrmx_1$, whereas the nearest point on $L_x$ to the left of the origin is closer than the corresponding intersection $\nrmx_2$. Thus, the probability of $\calE_4$ conditioned on $X_1$ and $X_2$ is
\begin{align}
\notag \P \left( \calE_4 | X_1, X_2\right) &= \P \left( D_1 > x_1, D_2 \leq x_2 | X_1, X_2 \right)\\
\notag &  =\P ( D_1 > x_1) \P (D_2 \leq x_2)\\
\label{eq:pe4_x1x2}
&= \exp(-\lambda_c x_1) \left( 1 - \exp(-\lambda_c x_2) \right) .
\end{align}

Similar to the previous case, conditioned on $\calE_4$, $X_1$, and $X_2$, the shortest path distance from the origin upon starting to the right is $R_1 = x_1 + Z_1$, where $Z_1$ is the shortest path distance from the intersection $\nrmx_1$. Note that the conditional CDF of $Z_1$ is the same as that of $W_1$ given in Lemma \ref{lem:cdfwm_e0}, i.e.,
\begin{align}
F_{Z_1}(z_1 | \calE_4 , x_1, x_2) = 
\begin{dcases}
F_{Z_1, 1} (z_1 | \calE_4 , x_1, x_2), \quad 0 \leq z_1 \leq x_2 ,\\
F_{Z_1, 2} (z_1 | \calE_4 , x_1, x_2), \quad  z_1 > x_2, 
\end{dcases}
\end{align}
where
\begin{align} 
F_{Z_1, 1} (z_1 | \calE_4 , x_1, x_2)  &= 1 - \exp\Bigg[-3 \lambda_c z_1 - 3 \lambda_l z_1   +\frac{3 \lambda_l}{2 \lambda_c} \left(1 - e^{-2 \lambda_c z_1}\right)\Bigg],
\end{align}
and 
\begin{align}
F_{Z_1, 2} (z_1 | \calE_4 , x_1, x_2)  = 1-\exp&\Bigg[- 3(\lambda_c + \lambda_l)z_1    +\frac{\lambda_l} {2\lambda_c} \big(3 + 2 e^{-2 \lambda_c x_2}- e^{-2 \lambda_c z_1} -4 e^{- \lambda_c (x_2+z_1)} \big)\Bigg].
\end{align}
However, the shortest path distance upon starting to the left is $D_2$. Therefore, the shortest path distance is $R_m = \min \{ R_1, D_2\}$ and its CDF can be computed as

\begin{align}
\notag F_{R_m} \left( r_m | \calE_4, x_1, x_2\right) & = 1 - \P \left( R_m > r_m | \calE_4, X_1, X_2\right)\\
\notag &= 1 - \P \left( \min\{R_1, D_2\} > r_m| \calE_4, X_1, X_2 \right) \\
\notag &= 1 - \P ( x_1 + Z_1 > r_m | \calE_4, X_1, X_2) \P ( D_2 > r_m | \calE_4, X_1, X_2 )\\
\notag &= 1 - \left( 1 - F_{Z_1}(r_m - x_1) \right)\frac{\left( F_{D_2} (x_2) - F_{D_2} (r_m)\right) \left(1 - F_{D_1}(x_1)\right) }{\P(\calE_4 | X_1, X_2)} \\
\label{eq:cdfrm_e4_x1x2}&= \begin{dcases}
\frac{1 - \exp(-\lambda_c r_m)}{1 - \exp(-\lambda_c x_2)}, & 0 \leq r_m \leq x_1, \\
1 - \frac{\left( 1 - F_{Z_1,1}(r_m - x_1) \right)\left( e^{-\lambda_c r_m} - e^{-\lambda_c x_2} \right)}{1 - \exp(-\lambda_c x_2)}, \ &x_1 \leq r_m < x_2,\\
1, & x_2 \leq r_m < \infty.
\end{dcases}
\end{align}

Using the results derived thus far, we now present the CDF of the overall shortest path distance $R_m$ in the following theorem.

\begin{theorem}\label{theorem:rm_gen_user}
	The CDF of the shortest path distance from the typical point of the MPLCP to its nearest neighbor in the sense of path distance is 
	\begin{align}
	\notag F_{R_m}(r_m) = &1 - e^{-2(\lambda_l+\lambda_c)r_m}  \\
	& \notag \hspace{2em} -2 \lambda_l e^{-(\lambda_l+\lambda_c)r_m}\int_0^{r_m} \big( 1-  F_{W_1,1}(r_m-x_1|\calE_3,x_1, x_2)\big) e^{-(\lambda_l+\lambda_c) x_1} \nrmd x_1 \\ 
	& \notag \hspace{2em} -2 \lambda_l^2 \int_{\frac{r_m}{2}}^{r_m}  e^{-(\lambda_l+\lambda_c) x_2} \big(1 - F_{W_2,1}(r_m-x_2| \calE_3, x_1, x_2)\big) \\
		&\notag \hspace{5em} \times \int_{r_m-x_2}^{x_2} \big(1 - F_{W_1,1}(r_m-x_1|\calE_3, x_1, x_2)\big) e^{-(\lambda_l+\lambda_c) x_1} \nrmd x_1 \nrmd x_2 \\
	& \notag \hspace{2em} -2 \lambda_l^2 \int_{\frac{r_m}{2}}^{r_m} e^{-(\lambda_l+\lambda_c) x_2} \int_0^{r_m - x_2}\big( 1 - F_{W_1,2}(r_m-x_1| \calE_3, x_1, x_2)\big)\\
		&\notag \hspace{5em} \times  \big(1 - F_{W_2, 2}(r_m-x_2|\calE_3, x_1,x_2)\big) e^{-(\lambda_l+\lambda_c) x_1} \nrmd x_1 \nrmd x_2 \\ 
	 & \notag \hspace{2em} -2 \lambda_l^2 \int_0^{\frac{r_m}{2}} e^{-(\lambda_l+\lambda_c) x_2} \int_0^{x_2} \big(1 - F_{W_1,2}(r_m-x_1| \calE_3, x_1, x_2) \big) \\
			\label{eq:cdfrm_gen}	& \hspace{5em} \times \big(1 - F_{W_2, 2}(r_m-x_2|\calE_3, x_1,x_2)\big) e^{-(\lambda_l+\lambda_c) x_1} \nrmd x_1 \nrmd x_2 .
	\end{align}	
%	\begin{align} \label{eq:cdfrm_gen1}
%	&F_{R_m}(r_m) = \sum_{i=1}^4 \int_0^\infty \int_0^{x_2} F_{R_m}(r_m | \calE_i, X_1, X_2) \P(\calE_i | X_1, X_2) f_{X_1, X_2}( x_1, x_2) \nrmd x_1 \nrmd x_2,
%	\end{align}
%	where the expressions for $F_{R_m}(r_m | \calE_i, X_1, X_2)$ are given in \eqref{eq:cdfrm_e1_x1x2}, \eqref{eq:cdfrm_e2_x1x2}, \eqref{eq:cdfrm_e3_x1x2}, \eqref{eq:cdfrm_e4_x1x2}, $\P(\calE_i | X_1, X_2)$ are given in \eqref{eq:pe1_x1x2}, \eqref{eq:pe2_x1x2}, \eqref{eq:pe3_x1x2}, \eqref{eq:pe4_x1x2}, and $f_{X_1, X_2}( x_1, x_2)$ is given in \eqref{eq:pdf_x1x2}, respectively.
\end{theorem}
\begin{proof}
	Using law of total probability, the CDF of $R_m$ conditioned on $X_1$ and $X_2$ can be computed as
	\begin{align}
	F_{R_m}(r_m | x_1, x_2) = \sum_{i=1}^4 F_{R_m}(r_m | \calE_i, X_1, X_2) \P(\calE_i | X_1, X_2), 
	\end{align}
	where the expressions for $F_{R_m}(r_m | \calE_i, X_1, X_2)$ are given in \eqref{eq:cdfrm_e1_x1x2}, \eqref{eq:cdfrm_e2_x1x2}, \eqref{eq:cdfrm_e3_x1x2}, \eqref{eq:cdfrm_e4_x1x2}, and $\P(\calE_i | X_1, X_2)$ are given in \eqref{eq:pe1_x1x2}, \eqref{eq:pe2_x1x2}, \eqref{eq:pe3_x1x2}, \eqref{eq:pe4_x1x2}, respectively.
	The overall CDF of $R_m$ can now be obtained by computing the expectation of the above expression w.r.t. $X_1$ and $X_2$ as
	\begin{align}
	F_{R_m}(r_m) = \int_0^\infty \int_0^{x_2} F_{R_m}(r_m | X_1, X_2) f_{X_1, X_2}( x_1, x_2) \nrmd x_1 \nrmd x_2,
	\end{align}
	where  $f_{X_1, X_2}( x_1, x_2)$ is given in \eqref{eq:pdf_x1x2}. Upon substituting the corresponding expressions and solving the resulting integrals, we obtain the expression given in \eqref{eq:cdfrm_gen}.
\end{proof}

\begin{figure*}
	\centering
	\begin{subfigure}[b]{0.235\textwidth}
		\centering
		\includegraphics[width=\textwidth]{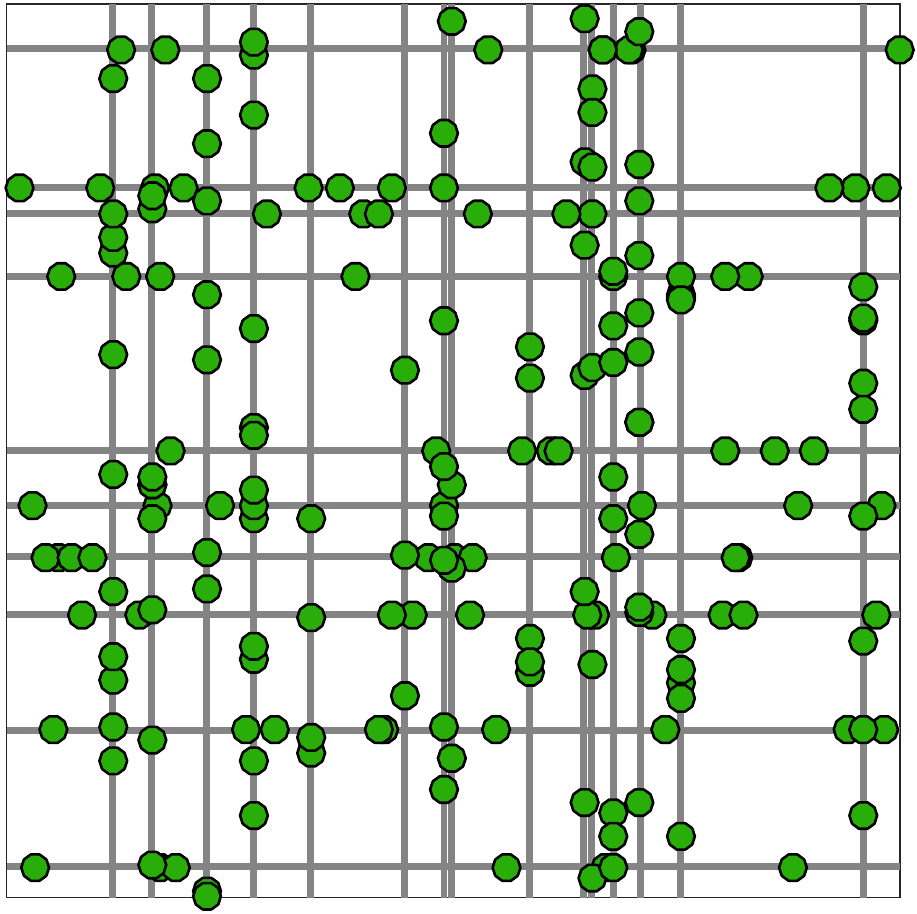}
		\caption{}%   
	\end{subfigure}
	\hfill
	\begin{subfigure}[b]{0.235\textwidth}  
		\centering 
		\includegraphics[width=\textwidth]{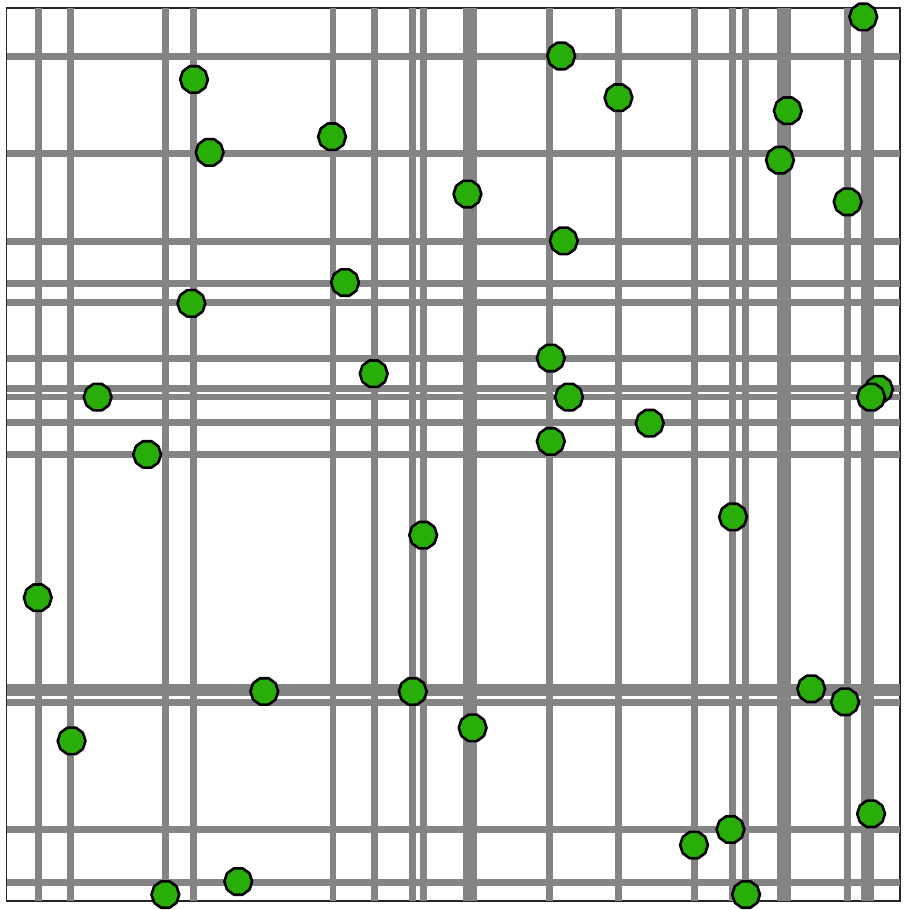}
		\caption{}   
	\end{subfigure}
	\hfill %%\vskip\baselineskip
	\begin{subfigure}[b]{0.235\textwidth}   
		\centering 
		\includegraphics[width=\textwidth]{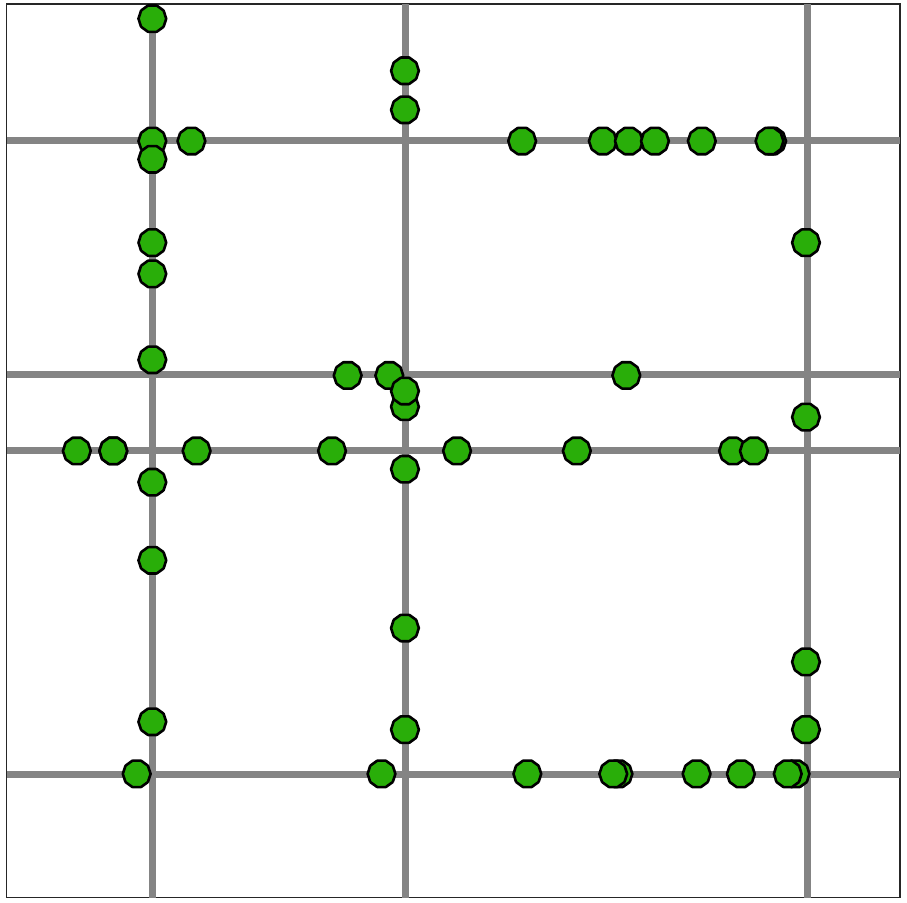}
		\caption{}
	\end{subfigure}
	\hfill
	\begin{subfigure}[b]{0.235\textwidth}   
		\centering 
		\includegraphics[width=\textwidth]{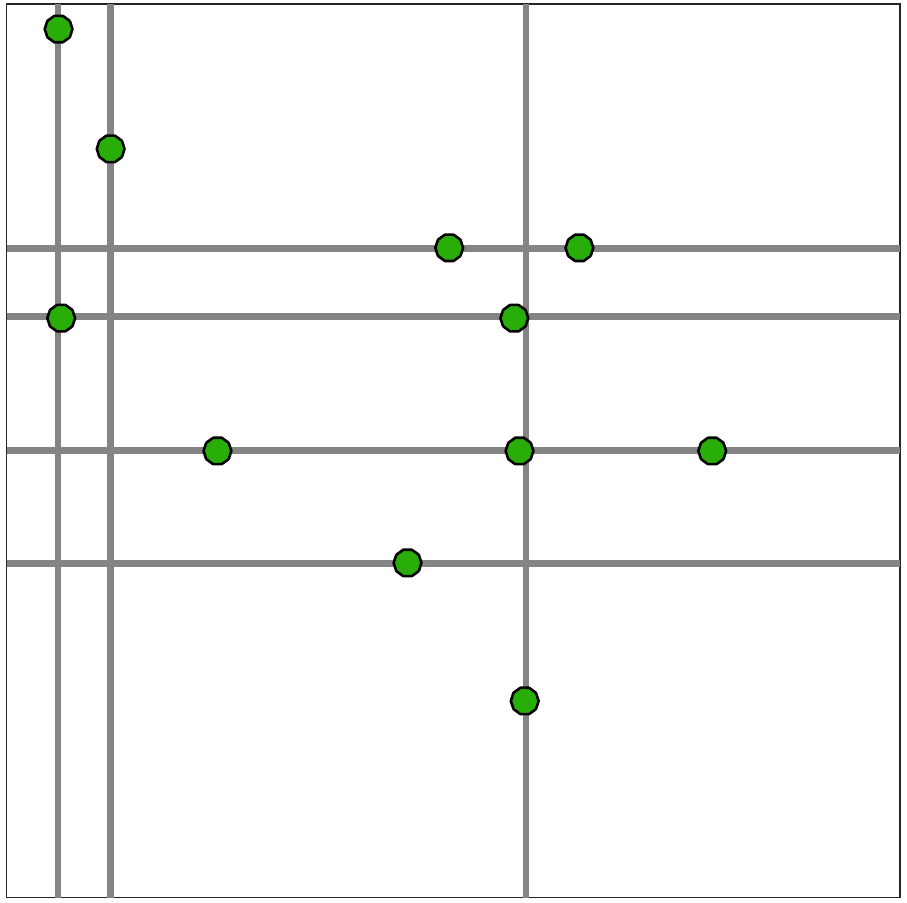}
		\caption{}
	\end{subfigure}
	\caption{The four different regimes of the spatial model: (a) Dense lines - dense points (DL-DP), (b) Dense lines - sparse points (DL-SP), (c) Sparse lines - dense points (SL-DP), and (d) Sparse lines - sparse points (SL-SP).}
	\label{fig:configs}
\end{figure*}

\begin{figure}
	\centering
	\includegraphics[width=.5\textwidth]{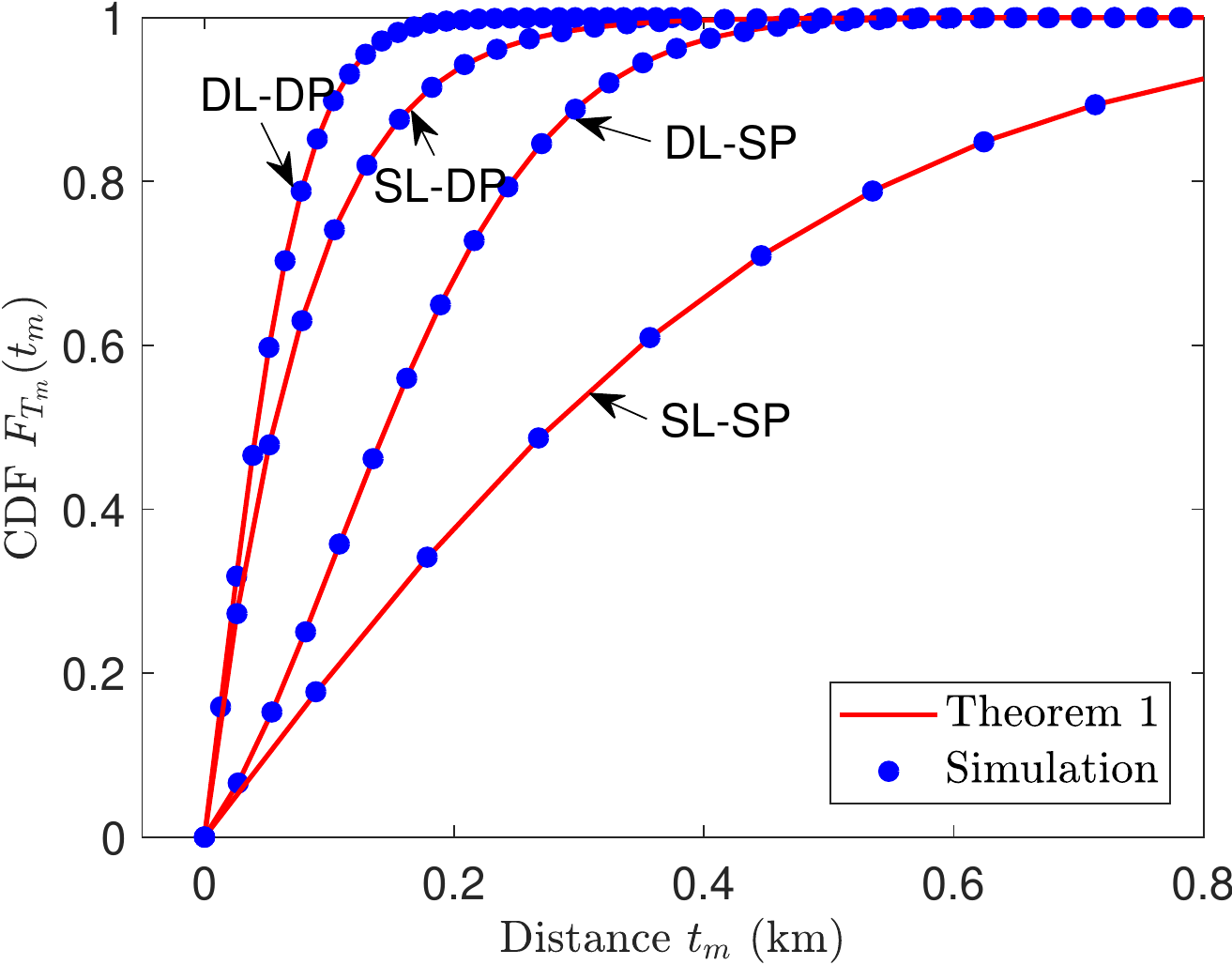}
	\caption{CDF of the shortest path distance for the {typical intersection case} for the four regimes: DL-DP ($\lambda_l = 10$ km$^{-1}$, $\lambda_c = 3$ points/km), SL-DP ($\lambda_l = 1$ km$^{-1}$, $\lambda_c = 3$ points/km), DL-SP ($\lambda_l = 10$ km$^{-1}$, $\lambda_c = 0.5$ points/km), and SL-SP ($\lambda_l = 1$ km$^{-1}$, $\lambda_c = 0.5$ points/km).}
	\label{fig:numcdf_int}
\end{figure}
\begin{figure}
	\centering
	\includegraphics[width=.5\textwidth]{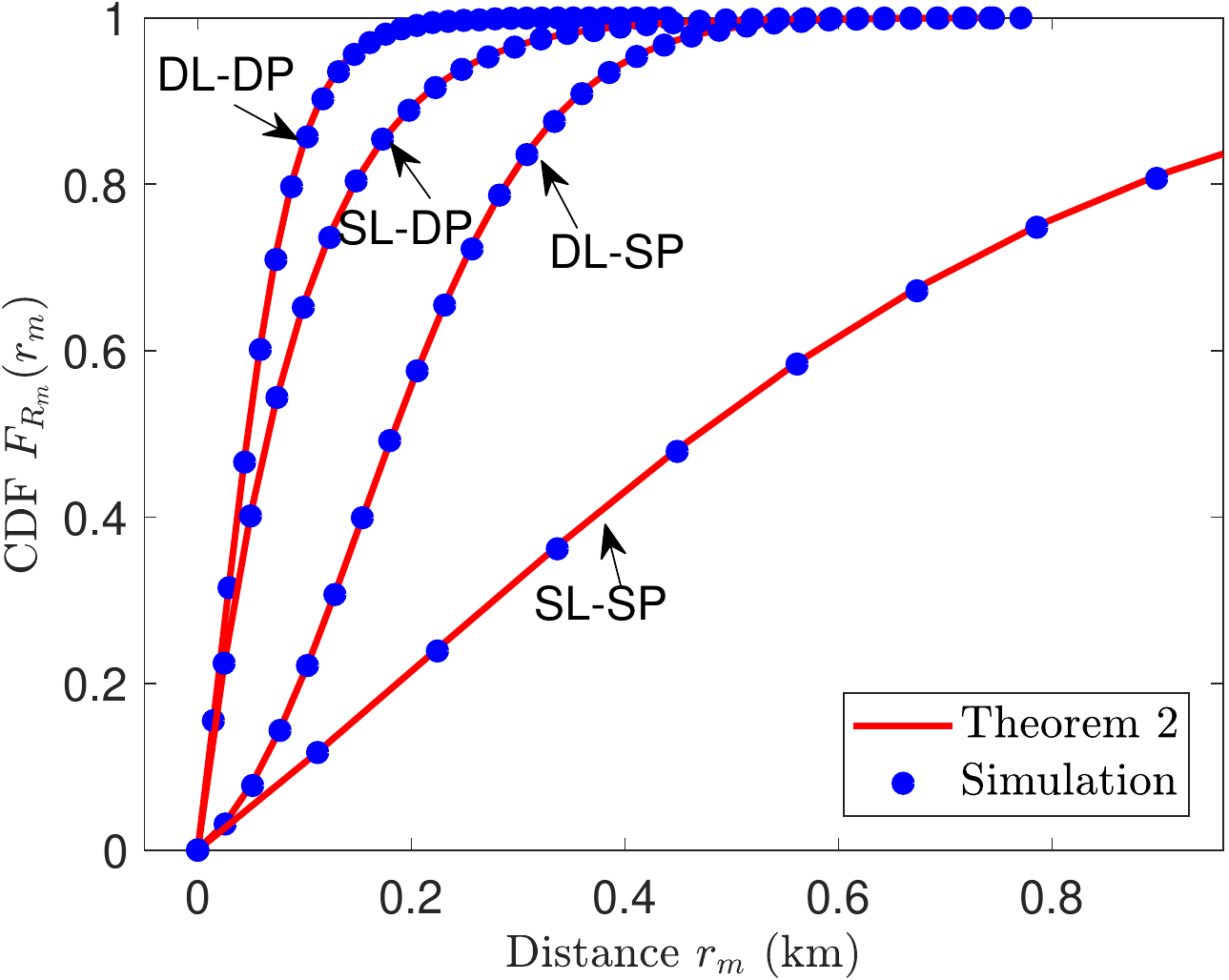}
	\caption{CDF of the shortest path distance for the typical point of the MPLCP for the four regimes: DL-DP ($\lambda_l = 10$ km$^{-1}$, $\lambda_c = 5$ points/km), SL-DP ($\lambda_l = 1$ km$^{-1}$, $\lambda_c = 5$ points/km), DL-SP ($\lambda_l = 10$ km$^{-1}$, $\lambda_c = 0.5$ points/km), and SL-SP ($\lambda_l = 1$ km$^{-1}$, $\lambda_c = 0.5$ points/km).}
	\label{fig:numcdf_gen}
\end{figure}

\section{Results and Discussion}

In this section, we will present the numerical results and discuss some of the applications of these results in transportation networks, infrastructure planning and wireless networks.

\subsection{Numerical results}
We first compute the empirical CDF of shortest path distances using Monte-Carlo simulations and compare it with the CDF obtained from the analytical expressions given in Theorems \ref{theorem:rm_int_user} and \ref{theorem:rm_gen_user}. In order to visualize the path distance distribution for different combinations of line and point densities, we evaluate the results under four broad regimes based on the densities of lines and points: (i) dense lines - dense points (DL-DP) corresponding to large values of $\lambda_l$ and $\lambda_c$, (ii) dense lines - sparse points (DL-SP) corresponding to large values of $\lambda_l$ and small values of $\lambda_c$, (iii) sparse lines - dense points (SL-DP) corresponding to small values of $\lambda_l$ and large values of $\lambda_c$, and (iv) sparse lines - sparse points (SL-SP) corresponding to small values of $\lambda_l$ and $\lambda_c$, as illustrated in Fig. \ref{fig:configs}. Note that Fig. \ref{fig:configs} is only for illustration purpose and the actual simulation values corresponding to these configurations are provided along with the results in Figs. \ref{fig:numcdf_int} and \ref{fig:numcdf_gen}. {As expected, the CDF obtained from the analytical expressions for both the cases match exactly with the corresponding empirical CDFs for all the configurations, as depicted in Figs. \ref{fig:numcdf_int} and \ref{fig:numcdf_gen}}.

\subsection{Applications to other areas}\label{sec:applications}

\subsubsection{Wireless communication} 
As we have mentioned in Section \ref{sec:intro}, the MPLCP can be used to model the locations of vehicular nodes and roadside units (RSUs) in a vehicular network and analyze key performance metrics such as coverage and rate by leveraging the Euclidean distance properties. However, in the case of millimeter wave communications in an urban environment, the high frequency radio signals suffer from severe attenuation upon propagating through the buildings and the dominant component of the signal is often the one that travels along the roads with diffractions around the corners at intersections \cite{HeathMmwave}. As a result, the analytical techniques developed in this paper can be leveraged to characterize the propagation delays and the received power of such signals. This is quite useful in deriving the power-delay profile of the wireless channel which is an important exercise in the performance analysis of wireless networks.

\subsubsection{Transportation systems and infrastructure planning}
In transportation networks, the spatial layout of roads can be modeled by MPLP and the various places of interest such as gas stations or charging stations for electric vehicles can be modeled by a MPLCP. Thus, the length of the shortest path studied in the paper can be viewed as the shortest distance that needs to be traveled by a vehicular user to reach the nearest destination of a certain type. Building further on the results presented in this paper, it is possible to analytically characterize the distance-dependent cost metrics that are of interest in transportation systems such as minimum travel time and fuel consumption. These results can be useful in characterizing the response time of medical or police personnel to arrive at the site of an emergency. Such analyses can also provide macroscopic insights into urban planning and design.

\section{Conclusion}
%Transportation systems have been mostly studied using spatial network models where various sites were modeled as nodes of a graph and the routes between those sites were modeled as edges of the graph. While these models are useful in studying the topological properties of the network, they do not accurately capture some of the geometric aspects such as continuity of streets. So, in this paper, we have modeled the spatial layout of roads by a MPLP and the locations of vehicles and facilities by independent and homogeneous 1D PPPs. 
In this paper, we focused on the analytical characterization of the shortest path distance in a stationary MPLCP. {In particular, for this spatial model, we derived the exact CDF of the shortest path distance to the nearest point of the MPLCP in the sense of path distance from the typical intersection of the MPLP and the typical point of the MPLCP.} We then discussed some useful applications of our results in wireless communication networks, transportation networks, infrastructure planning and personnel deployment.

This work has several extensions. First of all, the spatial model considered in the paper can be used to study other useful metrics such as route-length efficiency statistic which is defined as a function of the ratio of the shortest path distance between a pair of points to the corresponding Euclidean distance between those points \cite{aldous1}. While we have derived the results for a MPLCP, the analytical procedure can be extended to a PLCP. Also, the discussion on applications of our results in transportation, infrastructure planning, and wireless communication in Section \ref{sec:applications} could motivate future work in all these areas.

\section*{acknowledgements}
This work is supported by the US National Science Foundation (Grant IIS-1633363) and UK Engineering and Physical Sciences Research Council (Grant EP/N002458/1). The authors would like to thank an anonymous reviewer for the constructive feedback that helped in improving this paper. All the code required to reproduce the numerical results is available on GitHub \cite{code}.

\bibliographystyle{IEEEtran}
\bibliography{arxiv_v3.bbl}
%\input{Revised_manuscript.bbl}



%\begin{thebibliography}{}
%%
%% and use \bibitem to create references. Consult the Instructions
%% for authors for reference list style.
%\end{thebibliography}

\end{document}